\pdfoutput=1
\documentclass[11pt]{article}
\usepackage{geometry}
\usepackage{amsmath}

\usepackage{amsthm}
\usepackage{amssymb}

\usepackage{booktabs} % For formal tables
\usepackage[ruled]{algorithm2e} % For algorithms

\SetAlFnt{\small}
\SetAlCapFnt{\small}
\SetAlCapNameFnt{\small}
\SetAlCapHSkip{0pt}
\IncMargin{-\parindent}

\usepackage{amsmath}
\usepackage{amssymb}
\usepackage[tight]{subfigure}
\usepackage {tikz}
\usetikzlibrary {decorations.pathreplacing,hobby,backgrounds,fit,positioning,chains,shapes,calc,positioning,automata,arrows}
\definecolor {processblue}{cmyk}{0.96,0,0,0}
\definecolor {red}{cmyk}{0, 1, 0.838, 0.373}

\newcommand{\full}[1]{#1}
\newcommand{\wine}[1]{}
\newenvironment{proofsk}{\noindent{\bf Proof sketch:~~}}{\(\qed\)}

\full{
\newtheorem{theorem}{Theorem}[section]
\newtheorem{corollary}[theorem]{Corollary}
\newtheorem{lemma}[theorem]{Lemma}
\newtheorem{claim}[theorem]{Claim}
\newtheorem{definition}[theorem]{Definition}

\newtheorem{proposition}[theorem]{Proposition}

\newtheorem{observation}[theorem]{Observation}
}

\def\squarebox#1{\hbox to #1{\hfill\vbox to #1{\vfill}}}

\newcommand{\xhdr}[1]{\paragraph{{\textbf{#1}}.}}
\newcommand{\eps}{\varepsilon} 
\begin{document}
	\title {Mechanisms for Trading Durable Goods\footnote{Work partially supported by ISF grant 2167/19.}} 
	\author{Sigal Oren\thanks{Ben-Gurion University of the Negev. Email: sigal3@gmail.com.} 
		\and Oren Roth\thanks{Email: oren.roth@gmail.com. Most of the work was done while O. Roth was a student at Ben-Gurion University.} }
	
%		\title {Mechanisms for Trading Durable Goods\thanks{Work partially supported by ISF grant 2167/19. Missing proofs can be found in the full version which is available on arXiv.}} 
%	\author{Sigal Oren \and Oren Roth}
%	\authorrunning{S. Oren and O. Roth}
%	\institute{Ben-Gurion University of the Negav, Beer-Sheva 8410501, Israel  \\
%		\email{sigal3@gmail.com, oren.roth@gmail.com}}
	
	\maketitle
	\begin{abstract}
		We consider trading indivisible and easily transferable \emph{durable goods}, which are goods that an agent can receive, use, and trade again for a different good. This is often the case with books that can be read and later exchanged for unread ones. Other examples of such easily transferable durable goods include puzzles, video games and baby clothes. 

We introduce a model for the exchange of easily transferable durable goods. In our model, each agent owns a set of items and demands a different set of items. An agent is interested in receiving as many items as possible from his demand set. We consider mechanisms that exchange items in cycles in which each participating agent receives an item that he demands and gives an item that he owns.  We aim to develop mechanisms that have the following properties: they are \emph{efficient}, in the sense that they maximize the total number of items that agents receive from their demand set, they are \emph{strategyproof} (i.e., it is in the agents' best interest to report their preferences truthfully) and they run in \emph{polynomial time}. 

One challenge in developing mechanisms for our setting is that the supply and demand sets of the agents are updated after a trade cycle is executed. This makes constructing strategyproof mechanisms in our model significantly different from previous works, both technically and conceptually and requires developing new tools and techniques. We prove that simultaneously satisfying all desired properties is impossible and thus focus on studying the tradeoffs between these properties. To this end, we provide both approximation algorithms and impossibility results.
	\end{abstract}

	\section{Introduction}
	The sharing economy \cite{botsman2010s} puts on steroids the ancient idea of sharing physical assets. Instead of only sharing among friends, technology advancements facilitate the sharing of physical goods among strangers \cite{frenken2017putting}: rather than booking a hotel room, we often use Airbnb to live in someone's apartment, and instead of throwing away items we no longer use, we can exchange them for others that we do need. At the heart of the sharing economy is a desire to increase social efficiency by using underutilized resources. 

A perfect demonstration of the credo of the sharing economy is the efficient reallocation of \emph{durable} goods, like books, toys and sports gear. For example, suppose you have just finished reading the first Harry Potter book ``Harry Potter and the Philosopher's Stone''. You will probably be happy to exchange it in return for the second book in the series ``Harry Potter and the Chamber of Secrets''. In the offline world, for a swap to occur, we must have two individuals, each of whom is interested in the other's item. However, it is implausible that a person that read the second book in the Harry Potter series will be interested in the first book. Thus, such an exchange is more likely to occur as part of a larger cycle of exchanges which might be difficult to coordinate without a central mechanism. 

Online platforms (such as Swappy Books, Rehash Clothes, TradeMade and others) provide a central mechanism where each user can report the items he can give and the items he would like to receive. The platforms essentially provide infrastructure that reduces the search friction and allows to orchestrate complex exchanges. In this paper we focus on designing mechanisms for allocating easily transferable durable goods with the objective of maximizing the number of exchanges. A main challenge in designing such mechanisms is that the same item can be traded several times. This crucial difference from classic works on barter and trading (e.g., \cite{shapley1974cores,abraham2007clearing}) requires developing new algorithmic tools to ensure that these platforms will live up to their potential.  

\xhdr{A Model of Durable Goods} 
We consider a stylized model for exchanging easily transferable durable goods. Our model is based on the classic work of Shapley and Scarf \cite{shapley1974cores} for the house allocation problem. We consider a set $N$ ($|N|=n$) of agents and a set $M$ ($|M|=m$) of items. Each agent has a subset $D_i\subseteq M$ of items that he demands and a subset $S_i\subseteq M$ of items he owns. We make the simplifying assumption that each agent is willing to give any item from $S_i$ in return for any item from $D_i$. This means that the agent is indifferent between all the items that he demands and also between all the items that he owns. This is a reasonable assumption for books, for example, where books that the agent already read serve as a commodity that can be exchanged to get new desired books. Moreover, we assume that the agent is unwilling to receive an item that is not in his current demand set.\footnote{Formally, this can be modeled as part of the agent's utility as a large penalty that the agent exhibits for receiving an item not in the demand set or setting the utility to $0$ if the agent receives such an item.} This can be, for example, due to the physical or emotional burden of handling an unwanted item.  We model the demands and endowments of the agents as a directed bipartite graph $G=(N,M,E)$ in which there is a directed edge $(i,j)$ if agent $i$ demands item $j$ and a directed edge $(j,i)$ if agent $i$ owns item $j$. We refer to this graph as the \emph{trading graph}. 

We focus on mechanisms that execute exchanges according to cycles in the trading graph. In each exchange, every participating agent $i$ receives an item in $D_i$ and gives an item in $S_i$. The novelty of our model is that after agent $i$ received item $j\in D_i$ and used it, he can trade it later for another item from $D_i$. Our model is dynamic in the sense that after each step, the sets $D_i$ and $S_i$ are updated for each agent $i$. We refer to a sequence of cycles as an execution.

It is useful to go over an example to better understand the model. Consider the instance illustrated in Figure \ref{fig:model}. The set of agents is $N=\{a,b,c\}$ and the set of items is $M= \{x,y,z\}$. The set of items that agent $a$ demands is $D_a = \{x\}$ and the set of items that agent $a$ owns is $S_a=\{z\}$. For agent $b$: $D_b=\{y\}$ and $S_b=\{x\}$ and for agent $c$: $D_c=\{x,z\}$ and $S_c=\{y\}$. The trading graph of this instance is illustrated in Figure \ref{fig-ex-instance}. In the optimal execution we first execute the cycle $C_1=(b,y,c,x,b)$. That is, agent $b$ receives item $y$ from agent $c$ and agent $c$ receives item $x$ from agent $b$. Then, after agent $c$ receives item $x$ the edge $(c,x)$ is flipped. The graph after the execution of $C_1$ is illustrated in Figure \ref{fig-ex-instance-c1}. Now, we can execute the cycle $C_2=(a,x,c,z,a)$. As we see in Figure \ref{fig-ex-instance-c2}, after executing cycle $C_2$ there are no more cycles that we can execute. The social welfare of this execution is $4$ since overall it performs $4$ exchanges, the utilities of agents $a$ and $b$ is one since they received one item and the utility of agent $c$ is $2$ since he received two items from $D_c$. This is the optimal execution (i.e., the execution that performs the maximal possible number of exchanges) for this instance. This is in contrast to the best execution that can allocate each item at most once (i.e., a ``static'' execution) which consists of only $3$ exchanges (the cycle $(a,x,b,y,c,z,a)$ ).

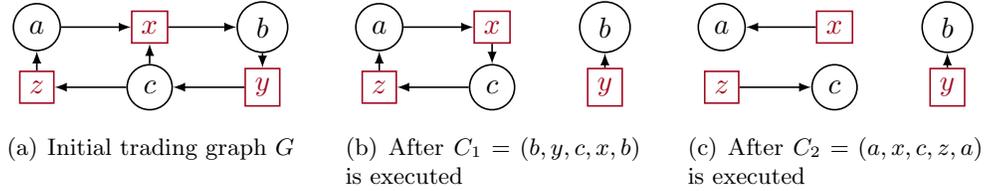
\begin{figure}[t]
	\begin{center}
		\subfigure[Initial trading graph $G$]{ \label{fig-ex-instance}
			\begin {tikzpicture}[-latex ,auto ,node distance =.8 cm and 1.5cm ,on grid ,
			semithick ,
			state/.style ={ circle ,
				draw , minimum width =.1 cm},
			item/.style = {red}]
			\node[state] (u){$a$};
			\node[draw,item] (v)[right=of u]{$x$};
			\node[state] (w) [right=of v] {$b$};
			\node[draw,item] (x) [below=of w] {$y$};
			\node[state] (y) [left =of x] {$c$};
			\node[draw,item] (z) [left =of y] {$z$};
			\path (u) edge [] node[below =0.15 cm] {} (v);
			\path (v) edge [] node[below =0.15 cm] {} (w);
			\path (w) edge [] node[below =0.15 cm] {} (x);
			\path (x) edge [] node[below =0.15 cm] {} (y);
			\path (y) edge [] node[below =0.15 cm] {} (v);
			\path (y) edge [] node[below =0.15 cm] {} (z);
			\path (z) edge [] node[below =0.15 cm] {} (u);

		\end{tikzpicture}
	}\hspace{5mm}
	\subfigure[After $C_1=(b,y,c,x,b)$ is executed]{ \label{fig-ex-instance-c1}
		\begin {tikzpicture}[-latex ,auto ,node distance =.8 cm and 1.5cm ,on grid ,
		semithick ,
		state/.style ={ circle ,
			draw , minimum width =.1 cm},
		item/.style = {red}]
		\node[state] (u){$a$};
		\node[draw,item] (v)[right=of u]{$x$};
		\node[state] (w) [right=of v] {$b$};
		\node[draw,item] (x) [below=of w] {$y$};
		\node[state] (y) [left =of x] {$c$};
		\node[draw,item] (z) [left =of y] {$z$};
		\path (u) edge [] node[below =0.15 cm] {} (v);
		\path (x) edge [] node[below =0.15 cm] {} (w);
		\path (v) edge [] node[below =0.15 cm] {} (y);
		\path (y) edge [] node[below =0.15 cm] {} (z);
		\path (z) edge [] node[below =0.15 cm] {} (u);
	\end{tikzpicture}
}\hspace{5mm}
\subfigure[After $C_2=(a,x,c,z,a)$ is executed]{  \label{fig-ex-instance-c2}
	\begin {tikzpicture}[-latex ,auto ,node distance = .8cm and 1.5cm ,on grid ,
	semithick ,
	state/.style ={ circle ,
		draw , minimum width =.1 cm},
	item/.style = {red}]
	\node[state] (u){$a$};
	\node[draw,item] (v)[right=of u]{$x$};
	\node[state] (w) [right=of v] {$b$};
	\node[draw,item] (x) [below=of w] {$y$};
	\node[state] (y) [left =of x] {$c$};
	\node[draw,item] (z) [left =of y] {$z$};
	\path (v) edge [] node[below =0.15 cm] {} (u);
	\path (x) edge [] node[below =0.15 cm] {} (w);
	\path (z) edge [] node[below =0.15 cm] {} (y);
\end{tikzpicture}

}
\caption{
	An example of an execution on a trading graph. Squares denote items and circles denote agents.
	\label{fig:model}
}
\end{center}
\vspace*{-0.2in}
\end{figure} 

\xhdr{Our Results} 
In this paper, we initiate the study of mechanisms that efficiently reallocate easily transferable durable goods. Due to the dynamic nature of this setting, designing such mechanisms is very challenging and requires developing new tools and techniques. Our goal is to develop algorithms that maximize the \emph{social welfare} which is defined as the total number of items that all agents receive from their demand sets.

 Furthermore, we take a mechanism design approach and assume that for each agent $i$, the demand set $D_i$ and the supply set $S_i$ are his private information. Thus, we would like our algorithms to be \emph{strategyproof} in the sense that each agent maximizes his utility (i.e., the number of items that he receives from his demand sets) by truthfully reporting his private information. Lastly, we would like our algorithms to run in \emph{polynomial time}. Ideally, we would like to develop algorithms that have all these properties. Unfortunately, we will prove that simultaneously obtaining all three properties is impossible. Thus we focus on studying the tradeoffs between them.

We begin by considering simple executions that reallocate each item at most once. We refer to such executions as \emph{static executions}. We show that an optimal static execution does not only provide a reasonable approximation to the social welfare of an optimal (dynamic) execution but can also be computed in polynomial time using a strategyproof algorithm. Formally, we show:

\begin{theorem}
	Let $l=\max_{i\in N } |D_i|$ be the maximum number of items that an agent demands. There exists a polynomial-time and strategyproof algorithm that computes an optimal static execution and provides an $l$-approximation to the social welfare of an optimal (possibly dynamic) execution. 
\end{theorem}

Our algorithm for computing the optimal static execution computes a maximum cycle cover of a graph by finding a maximum-weight perfect-matching in a corresponding bipartite graph. This general approach is similar to Abraham et al.  \cite{abraham2007clearing}; however, defining the appropriate bipartite graph for our model is more intricate as each agent may own and demand multiple items. Notably, in contrast to \cite{abraham2007clearing}, the algorithm that we devise is also strategyproof. This requires careful selection of the optimal execution that the algorithm outputs (there might be several allocations that provide the maximal welfare). Loosely speaking, we divide the proof that our algorithm is strategyproof into two parts. First, we show that an agent cannot increase his utility by not reporting some of the items that he demands. For this part of the proof, we take the edges of two maximum-weight perfect matchings for two instances: the instance in which agent $i$ reports his true demand set and the instance in which agent $i$ reports a subset of his demand. Using these edges, we construct two different matchings for the two instances. Since both pairs of matchings use the same set of edges, their sum of weights is identical.  Hence it is impossible that both original matchings are optimal for their corresponding instances. Then, we show that an agent cannot benefit from not reporting some items that he owns. We prove this by reversing the directions of all the edges in the graph and applying our previous result showing that an agent cannot benefit from not reporting some of the items in his demand set.

We then go on to study the limits of dynamic executions. First, we show that the approximation ratio achieved by our algorithm is close to the optimal approximation ratio achievable by any strategyproof algorithm. In particular, we show that any strategyproof algorithm cannot attain an approximation ratio better than $\approx \frac{l}{2}$ (where $l=\max_{i\in N } |D_i|$). Interestingly, we show that if we consider algorithms that are both strategyproof and always return a Pareto efficient allocation, the best achievable approximation ratio is $\Theta(n)$. Finally, we drop the requirement of strategyproofness and consider the computational question of computing an optimal execution. By constructing a careful reduction from the 3D-matching problem, we show that not only computing the optimal execution is NP-hard but it also cannot be approximated within some small constant unless P=NP. We note that the fact that the trading graph changes with the execution makes the reduction quite challenging. 
We leave open the question of closing the gap between the $l$-approximation ratio of our algorithm and the impossibility result showing that the problem cannot be approximated within a small constant.\footnote{In \full{Appendix \ref{sec-greedy}}\wine{the full version} we show that a greedy algorithm that sequentially finds an optimal static allocation cannot get an approximation ratio better than $l$.}

This work focuses on executions that are sequences of cycles in which each participating agent receives an item in his demand set. This restriction is a result of two main assumptions. First, we assume that an agent is unwilling to give an item without immediately receiving an item in return. The reasoning behind this is that the agent views the items he has as commodities used to get items that he is interested in. Thus, he does not want to lose such a commodity without immediately getting something in return. This can lead to a problem known as the ``double coincidence of wants'': Bob might demand some item that Alice has while Alice may not currently demand any item of Bob. The simple solution to this classic problem is to introduce some form of money. Thus, our results can be interpreted as demonstrating the necessity of money (not necessarily fiat money\footnote{\cite{kocherlakota1998money,akbarpour2020unpaired} suggest that ``memory'' could be used instead of money.}) when considering dynamic barter markets.  The second assumption we make is that an agent is not willing to give an item and receive in return an item that is not in his demand set. The rationale is that we assume that exchanging and perhaps storing an item that the agent is not interested in may have a nonnegligible physical or emotional cost associated with it.

	\xhdr{Related Literature}
Abbassi et al. \cite{abbassi2015exchange} consider a similar setting of barter networks with the main exception that each item may be allocated only once. This is comparable to our static executions but is very different from the more general dynamic executions. Abbassi et al. mainly focus on a different setting than ours in which the length of each trading cycle is bounded. For this setting, they present algorithms and hardness of approximation results for strategyproof mechanisms. For the setting in which the length of the trading cycles is unconstrained, they present a polynomial-time algorithm that computes an optimal static execution. It is important to note that, unlike our algorithm, this algorithm is not strategyproof.
	
The problem of computing an optimal static execution is very much related to the literature on matching. Our setting, in this respect, has two notable properties: 1) the agents may own and demand multiple items and 2) the agents have a high level of indifference in the sense that an agent only cares about the number of items from $D_i$ that he receives. In contrast, till recently, previous work did not consider situations that exhibit both of these properties. In \cite{sonmez1999strategy}, for example,  S\"onmez considered a setting in which each agent may own and demand multiple items. He showed that when each agent has a strict preference order over subsets of items, there is no individually rational, strategyproof, and Pareto-efficient mechanism.  Konishi et al. \cite{konishi} demonstrated that this impossibility result holds even in cases where there are only two different types of items (e.g., houses and cars) and the agents have strict preference order over these items. The high degree of indifference in our model allows us to escape these impossibility results. Till recently, most works in the matching literature featuring indifference considered only settings in which each agent owns and can receive exactly a single item (\cite{alcalde2011exchange,jaramillo2012difference,aziz2012housing,saban2013house}).

In a recent working paper, Manjunath and Westkamp \cite{manjunath2018strategy} considered a \emph{static} setting they refer to as \emph{Trichotomous Preferences}. Similar to our setting, each agent labels the item that he does not own as either desirable or undesirable. However, in contrast to our setting, the items that an agent \emph{owns} are also labeled as desirable or undesirable. The only undesirable items that an agent is willing to accept as part of a bundle are those in his initial endowment. The agents rank all the acceptable bundles according to the number of desirable items in them. The authors show that there is a computationally efficient mechanism that is individually rational, Pareto efficient and strategyproof in this setting. Similarly to \cite{abbassi2015exchange} the mechanism that Manjunath and Westkamp present is also based on a fixed ordering of the agents. In \cite{manjunath2018strategy} this often leads to executions that are Pareto efficient but may be far from optimal. A slightly different setting was considered in a working paper by Andersson et al. \cite{andersson2018organizing} in which each agent is endowed with multiple copies of an item that is unique only to him. Agents label other agents' items as desirable or undesirable and for desirable items, they have a cap on the number of units they are willing to receive. Andersson et al. present a mechanism that is individually rational, strategyproof and optimal. Notice that their result can not be applied to our static model since their model is too restrictive.

\xhdr{Paper Outline}
In Section \ref{sec:model} we formally define our model. In Section \ref{sec-opt-static} we present a strategyproof and computationally efficient algorithm that provides an $l$-approximation to the optimal execution by computing an optimal static execution. In Section \ref{sec:limitiation-dynamic} we discuss the limits of dynamic executions. 
	
	\section{Model}  \label{sec:model}
	We consider a set of agents $N$ ($|N|=n$) and a set of items $M$ ($|M|=m$). For each agent $i$ we 
denote the subset of the items that he owns by $S_i$ and the subset of items that he
demands by $D_i$. We assume that $(S_1,\ldots,S_n)$ is a partition of $M$ and that no player demands an item that he also owns (i.e., $S_i \cap D_i = \emptyset$).
 We denote an instance by $\mathcal G=(N,M,S_1,D_1,\ldots,S_n,D_n)$.
We denote by $l=\max_{i\in N} |D_i|$ the maximal number of items
demanded by a single agent. We model the agents' preferences using a directed bipartite
\emph{trading graph} $G=(N,M,E)$. For agent $i$ and item $j\in
D_i$ we have a  \emph{demand edge} $(i,j)\in E$. Similarly, for each item $j\in S_i$ we
	have a  \emph{supply edge} $(j,i)\in E$.

We only allow agents to exchange items within cycles such that each
participant $i$ gives an item that is currently in $S_i$ and receives an item that is currently in $D_i$.\footnote{
	As discussed in the introduction, we assume that the agents do
not want to hold on to items that are not in their demand set nor to give items without getting anything in return.} 
After a cycle is executed we update the demand and supply sets of all the agents accordingly and
update the graph by reversing all the demand edges of the cycle and removing all the supply
edges of the cycle. Formally:

\begin{definition} [Cycle Step]
	Let $G=(N,M,E)$ denote the current trading graph and let $C=(i_1,j_1,i_2,j_2,...,i_{k},j_{k},i_1)$ denote a cycle
	in $G$.
%	. %A cycle step
%	%$C=(i_1,j_1,i_2,j_2,...,i_{k},j_{k},i_1)$ can
%	be executed if all
%	$\forall 1\le t \le k, (j_t,i_{(t+1)(mod k)})\in E$. 
	After executing the cycle $C$, the edge set
	of the graph is updated to $E'=(E\setminus \{e\in C\})\cup \{(j_t,i_t) | (i_t,j_t) \in
	C\}$. The number of exchanges in a cycle step is $|C|= |\{(i,j) | i\in N,~ j\in M,~ (i,j)\in C \}|$.
\end{definition}
In this paper we study executions, these are sequences of cycles that obey the conditions we previously defined. Formally:
\begin{definition} Consider an instance $\mathcal G=(N,M,S_1,D_1,\ldots,S_n,D_n)$.
An execution $r=C_1,\ldots,C_k~$ for $\mathcal G$ is a sequence of cycles such that for each $1 \leq i \leq k$, $C_i$ is a cycle in $G_{i-1}$, where $G_{i-1}$ is the trading graph that results from executing cycles $C_1,\ldots,C_{i-1}$ sequentially on the original trading graph $G=(N,M,E)$ and $G_0=G$. 
\end{definition}

We denote the set of all executions by $R$ and the number of exchanges of an execution $r=C_1,\ldots,C_k$ by $|r| = \sum_{i=1}^k |C_i|$.
We denote the set of agents that participate in an execution $r$ by $N(r)$ and the set of demand edges that are used in an execution $r$ by $E(r)$.

We refer to a demand edge that was used in an execution and then was flipped and used as a supply
edge as a \emph{dynamic edge}. For example, the edge $(c,x)$ in the execution described in Figure \ref{fig:model} is a dynamic edge. The utility of an agent $i$ in
an execution $r$ is defined as the number of items from his demand set that $i$ received throughout the execution.
Furthermore, we assume that at each step agents are only willing to accept items that are currently in their demand set. Accepting an item currently not in their demand set results in setting their utility to $-n\cdot m$. Similarly, the utility of an agent that is asked to give an item that is not currently in his demand set is also $-n\cdot m$.   Formally,

\begin{definition} [Agent's utility]
Consider the instance $\mathcal G=(N,M,S_1,D_1,\ldots,S_n,D_n)$ and an execution $r$ for $\mathcal G$.
Let $A_i(r)$ denote the set of items that agent $i$ received in $r$. The utility of agent $i$ is: $u_i(r) = |A_i(r)|$, if for each trading cycle that $i$ participated in, the item that he received was in his \emph{current} demand set and the item that he gave was in his current supply set. Else, $u_i(r)=-n\cdot m$.
\end{definition}

We consider the standard objective function of maximizing the social welfare:
\begin{definition} [Social Welfare]
Consider the instance $\mathcal G=(N,M,S_1,D_1,\ldots,S_n,D_n)$ and an execution $r$ for $\mathcal G$. The social welfare of $r$ is $U(r)=\sum_{i\in N} u_i(r)$.
\end{definition}
In other words, as we only consider executions in which agents receive items in their demand set the social welfare equals to the total number of items that the agents received (i.e., $\sum_{i=1}^n |A_i(r)| $). We denote the execution maximizing the social welfare (i.e., the optimal execution) by $r_o$.

	\section{A Strategyproof $l$-approximation Algorithm} \label{sec-opt-static}
Most of the literature on exchange economics usually focuses on models in which each item can
be reallocated at most once. In our model we allow items to be reallocated several times. In our analysis, we will refer to executions that, as in the
traditional literature, allocate each item at most once as \emph{static executions} and executions that
can allocate each item more than once as \emph{dynamic executions}.

In this section we present a polynomial time algorithm for computing the optimal static execution. Recall that $l$ is the maximum number of items that a single agent demands. We show that our algorithm provides an $l$-approximation to the social welfare of an optimal \emph{dynamic} execution. Furthermore, we show that our algorithm is strategyproof -- each agent maximizes his utility by truthfully reporting the items that he demands and the items that he owns. 

\begin{proposition}\label{prop:static-l-approx}
In any instance $\mathcal G$, the social welfare of an optimal static execution is at least $\frac{1}{l}$ of the social welfare of
	an optimal (possibly dynamic) execution.
\end{proposition}
\begin{proof} 
Consider an instance $\mathcal G$ and let $r_o$ be an optimal execution. We will show that there exists a 
static execution $r_s$ such that each agent that received at least one item in $r_o$ will receive 
one item in $r_s$ (formally, $N(r_o) \subseteq N(r_s)$). As the maximal number of items that an 
agent may receive is $l$, this implies that the execution $r_s$ 
provides an $l$-approximation to the optimal (dynamic) execution. Thus, the optimal static execution also provides 
an $l$-approximation.

We now construct a static execution $r_s$ such that  $N(r_o) 
\subseteq N(r_s)$. Denote by $S(r_o)$ the set of supply edges that were used in 
	$r_o$. Let $G_{r_o}=(N,M, E(r_o)\cup S(r_o))$ denote the trading graph that includes all the edges 
	that were used in $r_o$. Since an execution is a sequence of cycles, the in-degree of each 
	node in $G_{r_o}$ is the same as its out-degree. Recall that a static execution can only allocate each item at most once, thus we essentially remove from $G_{r_o}$ the edges associated with items that were allocated more than once.  Formally, for each supply edge $(j,i)$ such that $i\in N$ and $j\in M$ that was not included in the initial trading graph $G$ (i.e., $(j,i) \in S(r_o)\setminus E $) we remove both the supply edge $(j,i)$ and the dynamic edge $(i,j)\in E$ that was flipped to create $(j,i)$. Denote the new graph by $G'_{r_o}$. 
Observe that the edges of $G'_{r_o}$ are a subset of $E$ and it is still the case that each node in $G'_{r_o}$ has the same out-degree and in-degree. Furthermore, since for each agent that received some items in $r_o$ the first supply 
	edge that was used is $(j,i) \in S(r_o) \cap E$ we have that if a node was not isolated in 
	$G_{r_o}$ it is still not isolated in $G'_{r_o}$. Thus, by Euler's theorem on the connected 
	components of $G'_{r_o}$ we have that there exists a cycle cover of all the edges in  
	$G'_{r_o}$. This cycle cover includes all the agents that received at least one item in $r_o$. 
	Therefore, we have that there exists a static execution $r_s$ for which $N(r_o) 
	\subseteq N(r_s)$ as required.
\end{proof}

Recall that for any instance $\mathcal G$ we denote by $r_s$ an optimal static execution for $\mathcal G$ and by $r_o$ an optimal execution for $\mathcal G$. We now show that this bound is tight:
\begin{claim}
For any $l\geq 1$ and $n> l$, there exists an instance $\mathcal G$ with $n$ agents such that  for each agent $i$ $|D_i|\leq l$ and $|r_s| = \frac{1}{l} |r_o|$. %$r_s$ is an optimal static execution for $\mathcal G$ and $r_o$ is an optimal execution for $\mathcal G$.
\end{claim}
\begin{proof}
	Consider an instance $\mathcal G$ with $l+1$ items ($M=\{1,\ldots,l+1\}$). Let $N_{l+1}=\{1,\ldots, l+1\}$ denote the set of the first $l+1$ agents. These are the only agents that own and demand items in $\mathcal G$.
	For each agent 
	$i \in N_{l+1}$ 
	we have that $S_i=\{i\}$ and $D_i = M \setminus \{i\}$. 
	The optimal execution executes $l$ cycle steps where in each step all agents in $N_{l+1}$
	give the item they own and receive a new item which they will swap in the next step. In this
	execution all the $(l+1)l$ demand edges are used. Note that only $l+1$ of them are not 
	dynamic edges. On 
	the
	other
	hand, in the optimal static execution each of the agents in $N_{l+1}$ will receive a single item. This is optimal as in any static execution the maximal number of items that an agent that owns a single item may receive is one. Hence, there is a gap of $l$ between the social welfare of the optimal execution and 
	the social welfare of the optimal static execution.
\end{proof}

\subsection{Computing an Optimal Static Execution.} 
\label{sec:static-poly}
We present an algorithm for computing an optimal static execution. Our algorithm computes a maximal cycle cover by finding a maximum weight perfect matching. This is similar to the algorithm of Abraham et al. \cite{abraham2007clearing} for the kidney exchange setting. However,  since in our setting the same agent may participate in more than one cycle, if he owns several items and demands several items, we consider edge-disjoint cycles whereas \cite{abraham2007clearing} considers node-disjoint cycles. To handle this difference we construct a new bipartite graph where we have two copies of each edge play as the vertices of the graph. Roughly speaking, each edge is 
connected to its copy with weight zero and to all the edges that are adjacent to it with weight 
$1$. 
Now, a perfect matching can define an execution in the following way: any edge that is matched to its copy does not take part in the execution and any edge that is matched to a different edge is included together with the edge it was matched to in some cycle in the execution. Notice that the maximum weight matching will maximize the number of edges that are not matched to their copies(and hence participate in the execution) as they are the only edges that have positive weight. We now formalize this construction:
\begin{theorem}\label{theorem:optimal-static-poly}
An optimal static execution $r$ can be computed in polynomial time ($O(|E|^3)$). 
\end{theorem}
\full{\begin{proof}}
\wine{\noindent \emph{Proof.}}
Given an instance $\mathcal G$ and its corresponding trading graph $G=(N,M,E)$ we construct a new undirected weighted bipartite graph, $H(G)=(E\times 
\{0\},E\times 
\{1\},E_H)$ such that:
$$
E_H = \underbrace{\{((e,0),(e,1))|  e\in E \}}_{E_1}\cup \underbrace{\{((e,0),(e',1))|e=(u,v), 
e'=(v,z) \in E \}}_{E_2}.
$$
We assign each edge in $E_1$ a weight of $0$ and each edge in 
$E_2$ a weight of $1$. We illustrate this construction in Figure \ref{fig:bipartite}.
\begin{figure}[t]
	\begin{center}
		\subfigure[A trading graph $G$]{
			\includegraphics[height=2cm]{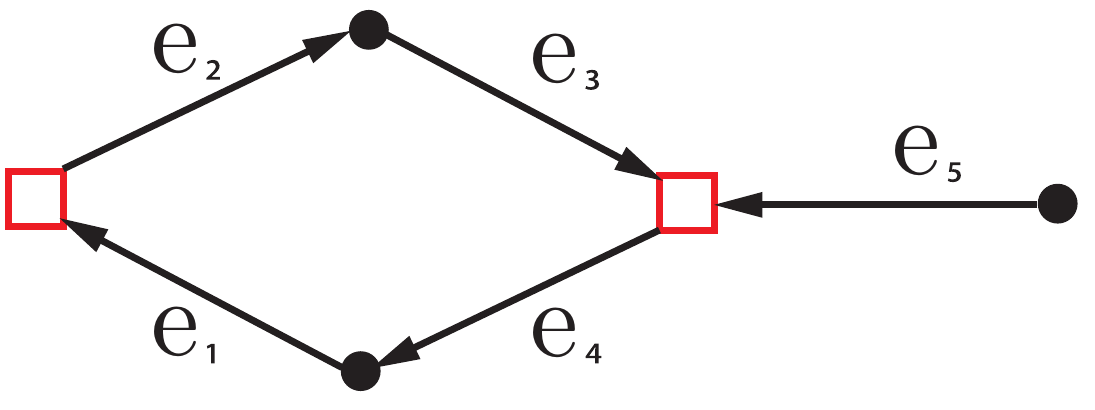}
			\label{fig:original-G}
		}\hspace{10mm}
		\subfigure[A bipartite edge graph $H(G)$]{
			\includegraphics[height=2cm]{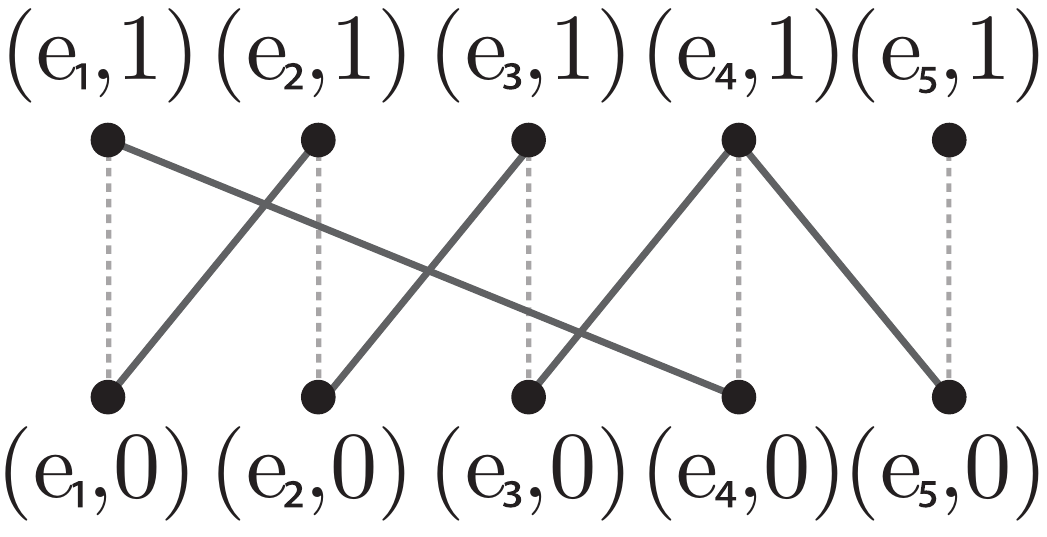}
			\label{fig:G'-insatnce}
		}
	\end{center}
	\vspace*{-0.2in}
	\caption{A trading graph $G$ and the corresponding bipartite edge graph $H(G)$. In the illustration of $H(G)$ dotted 
		edges have weight $0$ and solid edges weight $1$.}
	\label{fig:bipartite}
\end{figure} 
Note that we can construct $H(G)$ in polynomial time. After constructing the bipartite graph $H(G)$ 
we compute a maximum-weight perfect matching ${\mathcal M}$ of 
$H(G)$. This problem is known as the ``Assignment problem'' and can be solved in polynomial
time, for example using the Hungarian method \cite{assignment}. Note that a perfect 
matching of the graph $H(G)$ is guaranteed to exist since the edges of $E_1$ by themselves form a 
perfect matching. We establish the correctness of our algorithm by the following two lemmas\wine{ (The proofs of both lemmas can be found in the full version.)}. First\full{, in Lemma \ref{lemma:perfectmatching-static-1},} we construct 
from ${\mathcal M}$ a static execution $r$ such that $|r|$ equals half the weight of ${\mathcal M}$:
\begin{lemma}\label{lemma:perfectmatching-static-1}
	Consider a trading graph $G$. Given a perfect matching ${\mathcal M}$ for $H(G)$ of weight $x$ we can compute in polynomial time a corresponding
	static execution $r=C_1,\ldots,C_k$ for $G$ such that $|r|=x/2$.
\end{lemma}
\full{
\begin{proof}
	Let ${\mathcal M}$ be a maximum perfect matching in $H(G)$. Let $E_c$ be the edges of $G$ that 
	were 
	not matched to their copy in $H(G)$. Formally:
	\begin{align*}
		E_c =\{e\in E| \exists (e',1) \text{~such 
			that~} ((e,0),(e',1)) \in E_2 \cap {\mathcal M}  \}
	\end{align*}
	Consider the graph $\tilde G=(N,M,E_c)$. Since ${\mathcal M}$ is a perfect matching, we have that if 
	$(((a,x),0),((x,b),1))\in {\mathcal M} \cap E_2$ then $(((x,b),0),((b,y),1))\in {\mathcal M} \cap E_2$. This is because if $((x,b),0)$ 
	is not matched to $((x,b),1)$ it has to be matched to an edge that is adjacent to $(x,b)$ in $G$. From this we conclude 
	that 
	the in-degree and out-degree of each node in $\tilde G$ is the same. The reason for this is that for each incoming 
	edge there is a corresponding outgoing edge.
	
	Since the in-degree and out-degree of each node in $\tilde G$ is the same, we have that 
	each 
	connected component of $\tilde G$ has an Euler 
	cycle. By standard arguments we can partition each complex cycle to simple cycles. Thus, 
	we constructed a sequence of simple cycles covering all the edges of $E_c$. This is the static execution $r$.  Note that the weight of ${\mathcal M}$ is $x=|E_2 \cap {\mathcal M}|$ and the number of exchanges in the execution we 
	constructed is $|E_c|/2$. The reason for this is that $r$
	covers all 
	the edges in $E_c$ and we define the size of an execution as the number of 
	demand 
	edges which are half of the edges. Finally observe that by construction $|E_c|= |E_2 \cap {\mathcal M}|$ and thus $|r|=x/2$ as required.
\end{proof}
}

 To show that $r$ is an optimal static execution, in Lemma \ref{lemma:perfectmatching-static-2} 
 we prove that given an execution $r$ that makes $x$ exchanges we can construct a perfect matching of the graph $H(G)$ of weight $2x$.
\begin{lemma}\label{lemma:perfectmatching-static-2}
	If there exists a static execution $r=C_1,\ldots,C_k$ for the trading graph $G$ such that $|r|=x$, then a 
	perfect matching ${\mathcal M}$ for $H(G)$ of weight $2x$ exists.
\end{lemma}
\full{
\begin{proof}
	Given an execution $r=C_1,\ldots,C_k$, we construct a perfect matching as 
	follows: (1) For any consecutive edges  $e=(a,x)$ and $e'=(x,b)$ that are included in some cycle $C_i$ for $1\leq i \leq k$ we match 
	node $(e,0)$ to node $(e',1)$. (2) For any edge $e$ that does not appear in any cycle $C_i$ for $1\leq i \leq k$, we 
	match $(e,0)$ to $(e,1)$. The weight of this matching is exactly $2x$ as for 
	each exchange in $r$ we add one demand edge and one supply edge to the matching with 
	total weight of $2$. We now show this is a feasible and perfect matching. First, observe that 
	there are no unmatched nodes since each edge $e=(a,x)\in E$ either appears in some cycle, in this 
	case $((a,x),0)$ will be matched to a node $((x,b),1)$ and $((a,x),1)$ will be matched to a node $((y,a),0)$, or does not appear in any 
	cycle and in this case $(e,0)$ is matched to $(e,1)$. Furthermore, no node is matched more 
	than once since the cycles are edge disjoint and hence each edge either belongs to one cycle 
	or does not belong to any cycle.
\end{proof} \end{proof}
}

\subsection{A Strategyproof Algorithm for Computing an Optimal Static Execution}
\label{subsec:strategyproofness}

So far, we assumed that the agents truthfully report the items that they demand and own to the algorithm. In this section, we consider a mechanism design type of
question and ask if
indeed it is in the agents' best interest to report their true preferences. We denote the vector of the agents' reports by $\vec{x}'=(x'_1,x'_2,\ldots,x'_n)$, where
for each agent $i$: $x'_i=(D'_i,S'_i)$. We denote the trading graph
constructed by the reports as $G(\vec{x}')$. The utility of each agent depends on the
execution chosen by the mechanism. Recall that, roughly speaking, the utility of agent $i$ 
is the number of items that he received from his demand set $D_i$ in the execution. Formally, we denote by $A(\vec{x}')$
the execution computed by algorithm $A$ when it gets as input the reports vector $\vec{x}'$. With this notation, an algorithm $A$ is strategyproof (i.e., truthful) if and only if
$
\forall i, (D'_i,S'_i), \vec{x}'_{-i}$
$
~u_i(A((D_i,S_i),\vec{x}'_{-i})) \ge u_i(A((D'_i,S'_i),\vec{x}'_{-i}))
$
where $D_i$ and $S_i$ are the agent's private information.

In this section, we modify the $l$-approximation algorithm we presented in Section
\ref{sec:static-poly} for computing an optimal static execution to make it strategyproof. 
We note that the approximation ratio achieved by our algorithm is close to the optimal approximation ratio achievable by any strategyproof algorithm, as in Section
\ref{sec:limitiation-dynamic} we show that no strategyproof algorithm can guarantee an
approximation ratio better than $\frac{l+1}{2}$. This implies that the approximation ratio achieved by our algorithm is close to the optimal approximation ratio achievable by any strategyproof algorithm. We prove the following theorem:
\begin{theorem}\label{theorem:static-DSIC}
	There exists a strategyproof algorithm that computes an optimal static execution in poly-time.
\end{theorem}

In many cases there is no unique optimal static execution. In particular, often
there exists some agent $i$ that in one optimal execution receives more items than in another optimal execution. In a
strategyproof mechanism we need to make sure that such an agent cannot misreport the items that he demands or owns in order to get the mechanism to output an execution that is better for him. To handle this issue we apply a consistent tie-breaking role to select an optimal static execution. 

In particular, recall that in the algorithm described in Theorem
\ref{theorem:optimal-static-poly} we computed an optimal static execution by constructing a bipartite graph $H(G)$ in which the nodes of the graph are the edges of the trading graph $G$. Recall that in $H(G)$ the weight of each edge in $E_1$ is $0$ and
the weight of each edge in $E_2$ is $1$. We now slightly perturb the weights of the edges
in $H(G)$ to make sure that the algorithm breaks ties consistently between optimal static executions that
give different utilities to the same agent. To this end, we first define a complete order $\pi$ over all edges $(i,j)$ such that $i\in N$, $j\in M$. The order assigns each possible edge a distinct natural number between $1$ and $|N|\cdot |M|$. Next, we define a graph $H'(G)$ which is identical to $H(G)$ except that the weight of an edge $(((i,j),0),((j,k),1))\in E_2$ such that $i,k\in N$ and $j\in M$ is perturbed as follows:
$$w'(((i,j),0),((j,k),1))=1 + \underbrace{2^{-\pi((i,j))}}_{\eps_{(((i,j),0),((j,k),1))}}.$$
Similarly, the weight of an edge $(((j,i),0),((i,k),1))\in E_2$ such that $i\in N$ and $j,k\in M$ is:
\full{$}$w'(((j,i),0),((i,k),1))=1 + 2^{-\pi((i,j))} = w'(((i,j),0),((j,k),1)).$\full{$}

Observe that for any matching $\mathcal M$ we have that $w'(\mathcal M) = |E_2\cap \mathcal M| + \sum_{e\in E_2\cap \mathcal M} \eps_e$. It is not hard to see that for any matching $\mathcal M$ the sum of perturbations is less than $1$ (i.e., $\sum_{e\in \mathcal M } \eps_e < 1$). Thus, we have that
$\lfloor w'(\mathcal M) \rfloor = |E_2\cap \mathcal M|$. Together with the fact that $w(\mathcal M) = |E_2\cap \mathcal M|$ this implies the following claim:
\begin{claim} \label{clm:weighted-to-unweighted}
If  ${\mathcal {M}}'$ is a maximum weight perfect matching of the graph $H'(G)$, then ${\mathcal {M}}'$ is also a
maximum weight perfect matching of the graph $H(G)$.
\end{claim}
Next, we show that any execution that corresponds to some maximum weight perfect matching in $H'(G)$ gives each agent the same utility:

\begin{proposition}\label{prop:unique}
For any two maximum weight perfect matchings of the graph $H'(G)$: ${\mathcal M},{\mathcal M}'$ and any two executions  $r,r'$ that correspond to
${\mathcal M}$ and ${\mathcal M}'$ respectively, for any agent $i$, $A_i(r) = A_i(r')$. In other words, the utilities of all agents are identical in $r$ and $r'$.
\end{proposition}
\begin{proof} 
Assume towards contradiction that there exist two perfect-matchings of the graph $H'(G)$: ${\mathcal M},{\mathcal M}'$ such that $w'({\mathcal M}) = w'({\mathcal M}')$ but for executions $r$ and $r'$ that correspond to ${\mathcal M}$ and ${\mathcal M}'$ there exists an agent $i$ and an item $j$ such that $ j \in A_i(r)$ but $j \notin A_i(r')$. This implies that there exists an edge $(((i,j),0),((j,k),1))\in E_2\cap {\mathcal M} \setminus {\mathcal M}'$ where $k\in N$. In this case, $w'({\mathcal M})$ has a  $2^{-\pi((i,j))}$ term that will be missing from $w'({\mathcal M}')$. We note that this term has a unique exponent, as the only other option to achieve this term is by having an edge $(((j,i),0),((i,a),1))\in E_2\cap {\mathcal M}'$ for $a\in M$. However this is impossible since an agent cannot both own and demand the same item. Furthermore, since the exponent is unique this term cannot be derived by adding different perturbations. Thus, we conclude that  $w'({\mathcal M}) \neq w'({\mathcal M}')$ in contradiction to our assumption.
\end{proof}
\xhdr{Strategyproofness} %\label{sec:static-dsic}
Denote by $A_s$ the algorithm that computes an optimal execution by choosing a maximum weight perfect matching of the graph $H'(G)$ defined above. We now prove Theorem \ref{theorem:static-DSIC} and show that $A_s$ is strategyproof. To this end, we show that $\forall i,\vec{x}'_{-i},x'_i,$
$u_i(A_s(x_i,\vec{x}'_{-i})) \ge u_i(A_s(x_i',\vec{x}'_{-i}))$, where $x_i=(D_i,S_i)$ is a truthful report. We first consider reporting items that are not in $D_i$ or $S_i$ respectively. Observe that if as a result of this misreport agent $i$ receives an item which is not in $D_i$  or need to give an item which is not in $S_i$, then by definition his utility will be $0$. To show that the agent cannot benefit from such misreport when this is not the case, we prove a type of irrelevancy property. We show that an agent cannot improve his utility by not reporting items that the he did not receive or did not give\full{:} \wine{. In the full version we observe:}
\begin{observation}\label{observe:only-used-edges}
	Fix some agent $i$, (not necessarily truthful) report $x'_i=(D'_i,S'_i)$ and reports vector $\vec{x}'_{-i}$ for the rest of the agents. Let $r=A_s(x'_i,\vec{x}'_{-i})$. Denote by $B_i(r)$ the set of items that agent $i$ gave in $r$ and recall that $A_i(r)$ is the set of items that agent $i$ received in $r$. For any $\tilde D_i \subseteq D_i' - A_i(r)$ and $\tilde S_i \subseteq S'_i - B_i(r)$ we have that $u_i(A_s((D'_i - \tilde D_i, S'_i - \tilde S_i),\vec{x}'_{-i})) = u_i(A_s(x'_i,\vec{x}'_{-i}))$.
\end{observation}
\full{
\begin{proof}
	Let ${\mathcal M}',\tilde {\mathcal M}$ be maximum weight perfect matchings when agent $i$ reports $x'_i$ and $\tilde x_i=(D'_i - \tilde D_i, S'_i - \tilde S_i)$ respectively. Let $\tilde E_i = \{(i,j): j\in \tilde D_i \} \cup \{(j,i): j\in \tilde S_i \}$. Note that
	$w'({\mathcal M}') \geq w'(\tilde {\mathcal M})$ as adding more edges to the graph $H'(G)$ can only increase the weight of the maximum weight perfect matching. Furthermore, for each $e\in  \tilde E_i$, we have that $((e,0)(e,1)) \in {\mathcal M}'$. This implies that $\tilde {\mathcal M}' = {\mathcal M} - \tilde E_i$ is a feasible matching when agent $i$ reports $\tilde x_i$. Furthermore, since the weight of all the edges in $\tilde E_i$ is $0$ we have that $w'({\mathcal M}') =w'(\tilde {\mathcal M}') \leq w'(\tilde {\mathcal M})$. Thus, we conclude that $w'({\mathcal M}') = w'(\tilde {\mathcal M})$.
	
	Recall that when we compute a corresponding execution in Lemma \ref{lemma:perfectmatching-static-1} we only use edges in $E_2$. Since, ${\mathcal M'}\cap E_2 =\tilde {\mathcal M}' \cap E_2$, we conclude that agent $i$ will have the same utility in any execution that corresponds to ${\mathcal M'}$ and in any execution that corresponds to $\tilde {\mathcal M}'$.  Finally, since both $\tilde {\mathcal M}'$ and $\tilde {\mathcal M}$ are maximum weight perfect matchings when agent $i$ reports $\tilde x_i$, by Proposition~\ref{prop:unique} we have that the utility of agent $i$ in any execution that corresponds to $\tilde {\mathcal M}$ and any execution that corresponds to $\tilde {\mathcal M}'$ is the same. Thus, we conclude that $u_i(A_s((D'_i - \tilde D_i, S'_i - \tilde S_i),\vec{x'}_{-i})) = u_i(A_s(x_i,\vec{x'}_{-i}))$.
	\end{proof}
}

We conclude that:
\begin{corollary} \label{cor:adding-reports}
	Any agent $i$ cannot increase his utility by reporting that he demands an item $j\notin D_i$ or that he owns an item $j\notin S_i$.
\end{corollary}

The main part of the proof is showing that an agent cannot benefit from not reporting some of the items in his demand set.
In Proposition \ref{proposition:DSICRemove} we show that for any report of the items that the agent owns he cannot benefit from hiding items in his demand set. Then in Proposition \ref{proposition:DSICRemoveSupply} we apply Proposition \ref{proposition:DSICRemove} on an instance in which each agent switches between the items he demands and the items that he owns. We show that for any demand report the agent cannot benefit from hiding some of the items that he owns. The two propositions together with Corollary \ref{cor:adding-reports} complete the proof of Theorem \ref{theorem:static-DSIC} showing that for any agent $i$, demand and supply reports $D'_i$ and $S'_i$ and reports of the other agents $\vec{x}'_{-i}$ it hold that $u_i(A_s((D'_i , S'_i),\vec{x}'_{-i})) \le u_i(A_s((D_i , S_i),\vec{x}'_{-i}))$. First by Corollary \ref{cor:adding-reports} we have that an agent can never benefit from including in his demand report items that are not in his true demand set and including in his supply report items that are not in his true supply. Then, for $D'_i\subseteq D_i,S'_i\subseteq S_i$ Proposition \ref{proposition:DSICRemove} guarantees us that $u_i(A_s((D'_i , S'_i),\vec{x}'_{-i})) \le u_i(A_s((D_i , S'_i),\vec{x}'_{-i}))$.
Finally we use Proposition \ref{proposition:DSICRemoveSupply} to get that for any $S'_i \subseteq S_i$: $ u_i(A_s((D_i , S'_i),\vec{x}'_{-i})) \leq   u_i(A_s((D_i , S_i),\vec{x}'_{-i}))$ as required. We now state and discuss Proposition \ref{proposition:DSICRemove} and Proposition \ref{proposition:DSICRemoveSupply}.

\begin{proposition}\label{proposition:DSICRemove}
	For every agent $i$, supply report $S'_i \subseteq S_i$, reports of the other agents $\vec{x}'_{-i}$ and $X \subseteq D_i$ we have that $u_i(A_s((X, S'_i),\vec{x}'_{-i})) \leq u_i(A_s((D_i, S'_i),\vec{x}'_{-i}))$. 
\end{proposition}
\begin{proof}
	We define the function $f_i(X) = u_i(A_s((X,S'_i),\vec{x}'_{-i}))$ for every $X \subseteq D_i$. We claim that $f_i$ is a monotone set function for subsets of $D_i$. That is, $\forall X\subseteq D_i$ and $\forall Y\subseteq  X$ we have that  $f_i(X)\ge f_i(Y)$. Note that this concludes the proof of the proposition. Also note that in order to prove that $f_i$ is monotone it is sufficient to show that for any agent $i$, $X \subseteq D_i$ and $j \in X$, $f_i(X)\ge f_i(X - \{j\})$.
	
	Let $r_X=A_s((X,S'_i),\vec{x}'_{-i}))$ be the execution that the algorithm $A_s$ outputs when agent $i$ reports demand $X$. By Observation \ref{observe:only-used-edges}, for any $j\notin A_i(r_X)$ we have that $f_i(X)= f_i(X - \{j\})$ as required. Thus, for the rest of the proof we consider the case that $j\in A_i(r_X)$. 
	
	Let $r_{X-\{j\}}=A_s((X-\{j\},S'_i),\vec{x}'_{-i}))$. Denote by $Y=A_i(r_{X-\{j\}}) - A_i(r_X)$ the set of items that agent $i$ received when reporting $X-\{j\}$ but did not receive when reporting $X$. Assume towards contradiction that  $ f_i(X - \{j\}) > f_i(X)$. As $ f_i(X - \{j\}) \leq  f_i(X) - 1 + |Y|$, this implies that $|Y|\geq 2$. Let
	${\mathcal M}_X$  and ${\mathcal M}_{X - \{j\}}$ be the maximum weight perfect matchings that were computed as part of $A_s$ for demand reports $X$ and $X - \{j\}$ respectively (these are the matchings that are used to derive the executions $r_X$ and $r_{X - \{j\}}$). Roughly speaking, we will construct from the union of their edges two different matchings: ${\mathcal M}'_X$  and ${\mathcal M}'_{X - \{j\}}$ such that one is a valid perfect matching when agent $i$ reports $X$ and the other is a valid perfect matching when agent $i$ reports $X-\{j\}$. Since those matchings cover the same edges as ${\mathcal M}_X$  and ${\mathcal M}_{X - \{j\}}$ the sum of their weights is the same. This implies that one of the matchings ${\mathcal M}_X$ or ${\mathcal M}_{X - \{j\}}$ does not have the maximum weight. The crux of the proof is the following proposition, which we prove in \full{Appendix \ref{app-prop312}}\wine{the full version}:
	\begin{proposition} \label{prop:matching-properties}
		If $ f_i(X - \{j\}) > f_i(X)$, then, for the maximum weight perfect matchings ${\mathcal M}_X$ and ${\mathcal M}_{X-\{j\}}$ there exist matchings ${\mathcal M}'_X$ and ${\mathcal M}'_{X - \{j\}}$ such that:
		\begin{enumerate}
			\item ${\mathcal M}'_X$ is a valid perfect matching of the graph $H'(G((X,S'_i),\vec{x}'_{-i}))$
			(i.e., when agent $i$ reports $X$) and ${\mathcal M}'_{X - \{j\}}$ is a valid perfect matching of the graph $H'(G((X-\{j\},S'_i),\vec{x}'_{-i}))$ (i.e., when agent $i$ reports $X-\{j\}$).
			\item $w'({\mathcal M}_X) + w'({\mathcal M}_{X-\{j\}}) =w'({\mathcal M}'_X) +  w'({\mathcal M}'_{X - \{j\}})$
			\item $w'({\mathcal M}'_X)\ne w'({\mathcal M}_X)$.
		\end{enumerate}
	\end{proposition}
	Observe that the three statements of the proposition imply that either $w'({\mathcal M}'_X)> w'({\mathcal M}_X)$ or $w'({\mathcal M}'_{X - \{j\}}) > w'({\mathcal M}_{X-\{j\}})$. Since both ${\mathcal M}'_X$ and ${\mathcal M}_X$ are valid matchings of $H'(G((X,S'_i),\vec{x}'_{-i}))$ and both ${\mathcal M}'_{X - \{j\}}$ and ${\mathcal M}_{X-\{j\}}$  are valid matchings of $H'(G((X-\{j\},S'_i),\vec{x}'_{-i}))$ this in contradiction to the assumption that ${\mathcal M}_X$ and ${\mathcal M}_{X-\{j\}}$ are maximum weight perfect matchings. Thus, we conclude that $f_i(X - \{j\}) \leq f_i(X)$.
\end{proof}

Next, we show that an agent maximizes his utility by truthfully reporting all the items that he owns. In the following proofs we will compare the utility of agents in two different instances. For this purpose we use the notation
$u_i^{\mathcal G}(r)$ for the utility of agent $i$ in execution $r$ of instance $\mathcal G$.

\begin{proposition}\label{proposition:DSICRemoveSupply}
	For any agent $i$, demand report $D'_i \subseteq D_i$, reports of the other agents $\vec{x}'_{-i}$ and $S'_i \subset S_i$ we have that $u_i(A_s((D'_i, S'_i),\vec{x}'_{-i})) \leq u_i(A_s((D'_i, S_i),\vec{x}'_{-i}))$.
\end{proposition}
\begin{proof}
	Denote the original instance of the problem by $\mathcal G$. We define the reversed instance $\bar {\mathcal G}$, in this instance the demand of each agent $a\in N$ is $\bar D_a=S_a$ and his supply is $\bar S_a = D_a$. For a vector of reports $\vec{x}'$ such that ${x}'_a=(D'_a,S'_a)$ for every agent $a\in N$ we define the reversed reports vector $\vec{\bar x}'$ such that for each agent $a$, $\bar x'_a=(S'_a,D'_a)$. 
	
	By Proposition \ref{proposition:DSICRemove} we have that for the instance $\bar {\mathcal G}$ and any reports vector of the other agents $\vec{\bar x}'_{-i}$, for any $S'_i \subset \bar D_i$ and any $D'_i \subseteq \bar S_i$:  $u^{\bar {\mathcal G}}_i(A_s((S'_i, D'_i),\vec{\bar x}'_{-i})) \leq u^{\bar {\mathcal G}}_i(A_s(( S_i, D'_i),\vec{\bar x}'_{-i}))$. To prove the proposition we will show in Claim \ref{clm:reverse} below that $u^{\bar {\mathcal G}}_i(A_s((S'_i, D'_i),\vec{\bar x}'_{-i})) = u^{ {\mathcal G}}_i(A_s((D'_i,S'_i), \vec{x}'_{-i}))$ and that $u^{\bar {\mathcal G}}_i(A_s(( S_i, D'_i),\vec{\bar x}'_{-i})) = u^{{\mathcal G}}_i(A_s(( D'_i, {S_i}), \vec{x}'_{-i}))$.%Observe that to prove the claim it is not enough to define for each execution in one instance a reversed execution for the reversed instance that performs the same number of exchanges as the fact that such an execution exists does not guarantee that the algorithm will return it. 
\end{proof}

We now observe the strong symmetry between an instance of our game and the reversed instance in which each agent swaps between the items he receives and the items he demands\full{:}\wine{. In the full version we prove:}
\begin{claim} \label{clm:reverse}
	For every agent $i$, reports vector $\vec{x}'$ for $\mathcal G$ and a reversed report vector $\vec{\bar x}'$ for $\bar{\mathcal G}$, we have that  $u^{{\mathcal G}}_i(A_s( \vec x '))=u^{\bar {\mathcal G}}_i(A_s(\vec{\bar x}'))$. 
\end{claim}
\full{
\begin{proof}
	Let $\mathcal M$ denote the maximum weight perfect matching in $H'(G(\vec{x}'))$ that was used in the run of the algorithm $A_s$. Let $\bar {\mathcal M} = \{ (((k,j),0),((j,i),1)):  (((i,j),0),((j,k),1))\in {\mathcal M} \}$. It is easy to see that $\bar {\mathcal M}$ is a perfect matching of the graph $H'(G(\vec{\bar x}'))$. Moreover, we claim that $w'(\bar {\mathcal M}) = w'({\mathcal M})$. To see why the two matchings have the same weight note that by construction ${\mathcal M}\cap E_1=\bar {\mathcal M} \cap E_1$ and that for every  edge $(((i,j),0),((j,k),1)) \in {\mathcal M}\cap E_2$  there exists a single edge $( ((h,i),0),((i,j),1)) \in {\mathcal M}\cap E_2$. By construction this implies that $( ((j,i),0),((i,h),1)) \in \bar {\mathcal M}\cap E_2$ and by the definition of the weights we have that $w'((((i,j),0),((j,k),1)))=w'(((j,i),0),((i,h),1))$. Thus we conclude that $w'(\bar {\mathcal M}) = w'({\mathcal M})$. Observe that the same argument shows that for $\bar{\mathcal M}'$ the perfect matching in $H'(G(\vec{\bar x}'))$ that was used in the algorithm $A_s$ and for the perfect matching $\mathcal M'$ for $H'(G(\vec{x}'))$ that is defined as above by reversing that matches of $\bar{\mathcal M}'$ we have that $w'(\bar {\mathcal M}') = w'({\mathcal M}')$. Thus, we conclude that $\bar {\mathcal M}$ is a maximum weight perfect matching of the graph $H'(G(\vec{\bar x}'))$. By Proposition \ref{prop:unique} this implies that the utility of agent $i$ in the execution that $A_s$ returns (which corresponds to $\bar {\mathcal M}'$ ) is the same as his utility from any execution that corresponds to $\bar {\mathcal M}$. Thus, to complete the proof it suffices to show the utility of all agents is the same in an execution $r$ that corresponds to ${\mathcal M}$ and an execution $\bar r$ that corresponds to $\bar {\mathcal M}$. Note that the utility of every agent $i$ in an execution $r$ that corresponds to ${\mathcal M}$ equals the number of items that he received: $|\{ j| \exists (j,k) \in E,~ (((i,j),0),((j,k),1)) \in {\mathcal M}\}|$.  Furthermore, we have that  
	$$\{ j| \exists (j,k) \in E,~ (((i,j),0),((j,k),1)) \in {\mathcal M}\} = \{ j| \exists (k,j)\in E,~ (((k,j),0),((j,i),1)) \in \bar {\mathcal M}\}.$$
	Notice that $|  \{ j| \exists (k,j)\in E,~ (((k,j),0),((j,i),1)) \in \bar {\mathcal M}\}|$
	is the number of items that agent $i$ gives in any execution $\bar r$ that corresponds to $\bar {\mathcal M}$. Since all exchanges are executed in cycles the number of items that a agent gives equals to the number of items that a agent receives and we conclude that the utility of agent $i$ in any execution $\bar r$ is  $| \{ j| \exists (k,j)\in E,~ (((k,j),0),((j,i),1)) \in \bar {\mathcal M}\} |$ as required.	
\end{proof}
}

	\section{Limitations of Dynamic Executions}\label{sec:limitiation-dynamic}
In Section \ref{subsec:strategyproofness} we showed that the efficient algorithm that computes the optimal static execution is both strategyproof and provides an $l$-approximation. In this section we prove that the best approximation ratio achievable by a strategyproof algorithm is $\frac{l+1}{2}$. Then, we consider the problem of finding the optimal execution from a strictly computational perspective and prove that unless P=NP the problem cannot be approximated within some small constant. 

 \begin{theorem}\label{THM:no-approx}
 There is no strategyproof algorithm which gives better approximation than $\frac{l+1}{2}$.
 \end{theorem}
\full{\begin{proof}} \wine{\begin{proofsk}}
Consider the following instance $\mathcal G$: the set of agents contains two subsets of cardinality $l$: $N=\{ i_1,i_2,\ldots, i_l \}$ and $N'=\{ i'_1,i'_2,\ldots, i'_l \}$. The set of items contains two subsets of cardinality $l$: $M = \{ j_1,j_2,\ldots, j_l \}$ and $M' = \{ j'_1,j'_2,\ldots, j'_l \}$. For $1 \leq k \leq l$ the demand of agent $i_k$ is $D_{i_k} = M$ and his supply is $S_{i_k} =\{j'_k\}$. For $1 \leq k \leq l$ the demand of agent $i'_k$ is $D_{i'_k} = \{ j'_k\}$ and his supply is $S_{i'_k} = \{ j_k\}$. The instance $\mathcal G$ also includes $\frac{l+1}{2}(l^2+l) \cdot l$ extra agents that are partitioned into $l$ groups. There are also extra items such that the demand and supply of each group of  $\frac{l+1}{2}(l^2+l)$ agents creates a path that ends in an item in $M'$. \full{Formally, for each item $j'_k \in M'$, we have a set of agents $N_k=\{ i_k^t: 1\le t\le \frac{l+1}{2}(l^2+l) \}$ and a set of items $M_k=\{ j_k^t: 1\le t\le \frac{l+1}{2}(l^2+l) \}$. Agent $i_k^t$ has item $j_k^t$ and demands item $j_k^{t+1}$ for $1\le t\le \frac{l+1}{2}(l^2+l)-1$, and agent $i_k^{\frac{l+1}{2}(l^2+l)}$ has item $j_k^{\frac{l+1}{2}(l^2+l)}$ and demands item $j'_k\in M'$.}  We illustrate the corresponding trading graph in Figure ~\ref{fig:dsic-no-approximation-better-than-l}. 

 \begin{figure}[t]
 	\centering
 	\includegraphics[width=.4\textwidth]{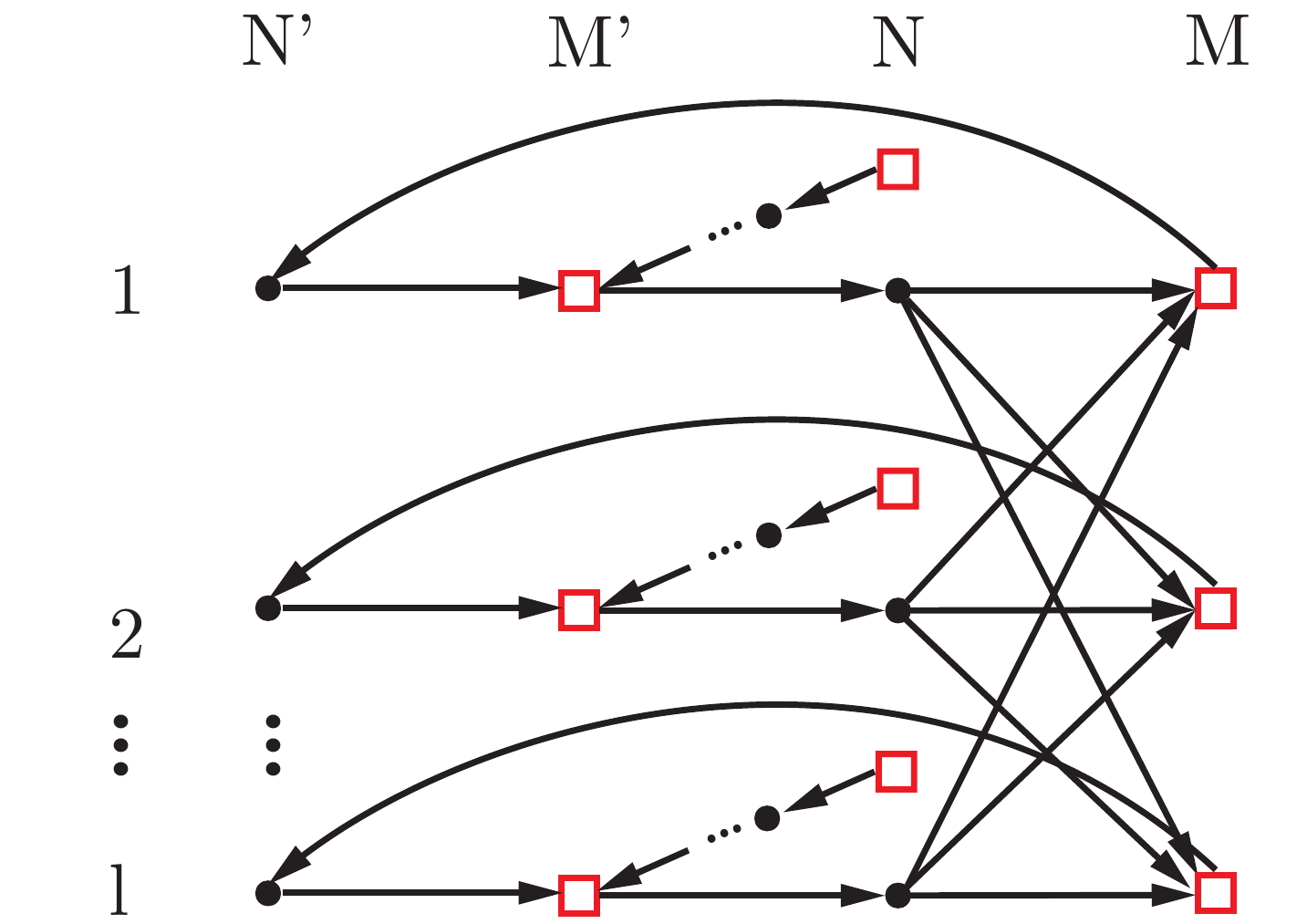}
 	\caption{ The trading graph for an instance showing that no strategyproof algorithm can attain an approximation ratio better than $\frac{l+1}{2}$.}
 	\label{fig:dsic-no-approximation-better-than-l}
 \end{figure}

 The optimal execution for the instance $\mathcal G$ first executes the $l$ cycles in which each pair of agents $i_k$ and $i'_k$ swap items $j'_k$ and $j_{k}$ between them. There are $l$ such cycles and in each cycle there are 2 exchanges. Then, it executes $l-1$ cycles with the items of $M$. The number of exchanges in each cycle is $l$ so the total number of exchanges is $ (l-1)\cdot l + 2l= l^2+l$. \full{This is the optimal execution as all the agents in $N'$ and $N$ receive all the items that they demand while the agents in $N_k$ for $1\le k\le l$ cannot receive any item since they do not take part in any cycle in the graph.
 }Assume towards contradiction that there exists an algorithm that achieves an approximation ratio $\alpha < \frac{l+1}{2}$. 
 %Consider an instance such that  $d=\lceil \frac{\alpha}{l-\alpha } \rceil$. 
 Note that in $\mathcal G$ such an algorithm must allocate at least one agent two or more items. \full{The reason is that the number of agents in $\mathcal G$ that can receive an item is $2\cdot l$ and hence the number of exchanges in an execution that only allocates each agent a single item is $2\cdot l$. Thus the approximation ratio achieved by an algorithm that allocates each agent at most single item is at least $ \frac{l^2+l}{2l} =\frac{l+1}{2}$. 
 }

 We conclude that there exists an agent $i_k \in N$ that is allocated by the algorithm in the instance $\mathcal G$ at least two items.  Now consider an instance $\mathcal G'$ which is identical to $\mathcal G$ except that agent $i_k$ also demands item $j_k^{1}$. This means that the trading graph now has a giant cycle of size $\frac{l+1}{2}(l^2+l)$ and since the algorithm guarantees an approximation ratio better than $\frac{l+1}{2}$ it has to execute this cycle. \wine{In this case $i_k$ can only participate in the giant cycle and hence only gets a single item. Thus, agent $i_k$ can increase his utility by not reporting $j_k^1$.}\full{Notice that if the giant cycle was the first cycle executed that includes agent $i_k$, then agent $i_k$ cannot participate in any more exchanges since now the only item he owns is $j_k^{1}$ and no other agent demands this item. If in the first cycle agent $i_k$ gave item $j'_k\in M'$ to agent $i'_k \in N'$, then the giant cycle cannot be executed since now item $j'_k\in M'$ is in the possession of agent $i'_k \in N'$ that does not demand any other item.
  
 Thus,  in any execution that executes the giant cycle the utility of agent $i_k$ is $1$. While the utility of agent $i_k$  in the instance $\mathcal G$ is at least $2$. Hence agent $i_k$ can increase his utility by not reporting item $j_k^1$. Thus, any algorithm that guarantees approximation ratio $\alpha < \frac{l+1}{2}$ is not strategyproof.}
\full{\end{proof}} \wine{\end{proofsk}}
 
\subsection{Pareto Efficiency and Strategyproofness}
An execution $r$ is Pareto efficient if for any other execution $r'$ there exists an agent $i$ such that $u_i(r') < u_i(r)$. We leave open the question of whether there exists a strategyproof algorithm that always returns a Pareto efficient execution. In any case, we show that even if such an algorithm exists, its performance is quite poor:
\begin{proposition} \label{prop-pareto}
Any algorithm that is strategyproof and returns a Pareto efficient execution cannot guarantee an approximation ratio better than $\Theta(n)$. 
\end{proposition}
\begin{proof}
Consider the following $n$-agent instance $\mathcal G$. In this instance we have 3 agents $i_1,i_2,i_3$ such that for agent $i_1$, $D_{i_1} =\{ a,b \}$ and $S_{i_1}=\{c\}$. For agent $i_2$,  $D_{i_2} = \{ b\}$ and $S_{i_2}= \{a\}$. For agent $i_3$, $D_{i_3}= \{c\}$ and $S_{i_3}=\{b\}$. The instance also include a sequence of $n-3$ agents that in the trading graph take part in a long path that starts from item $p'_1$ and end in item $c$: $P=(p'_1,p_1,p'_2,\ldots,p_{n-3},c)$. In figure~\ref{fig:lie-by-hide} we illustrate the trading graph for the instance $\mathcal G$.
\begin{figure}[t]
	\centering
	\includegraphics[width=.6\textwidth]{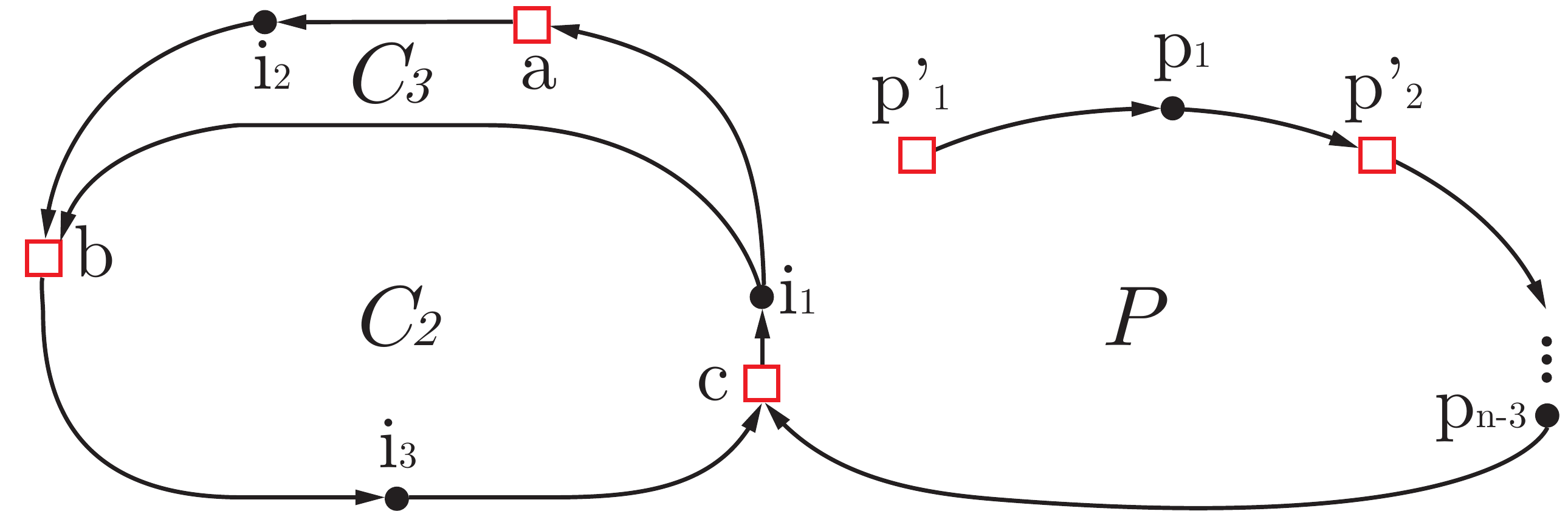}
	\caption{The trading graph for instance $\mathcal G$ in the proof of Proposition \ref{prop-pareto}.}
%	\caption{The trading graph include a cycle $C_2$ which after executing it a new cycle $C_3$ can be executed. If agent $i_1$ demands also item $p'_1$ then there is another feasible cycle, $C_1$, which cannot be executed along $C_2$.}
		\label{fig:lie-by-hide}
\end{figure}

In this instance the only (dynamic) Paerto efficient execution first executes the cycle $C_2=(i_1,b,i_3,c,i_1)$ and then executes the cycle $C_3=(i_1,a,i_2,b,i_1)$. The utility of agent $i_1$ in this execution is $2$. Consider the case that agent $i_1$ also demands item $p'_1$. Now, the graph has a giant cycle $C_1 = (p'_1,p_1,p'_2,\ldots,p_{n-3},c,i_1,p'_1)$ that includes $n-2$ agents and $n-2$ items. Note that it is impossible to execute both cycles $C_1$ and $C_2$ since in $C_2$ agent $p_{n-3}$ receives item $c$ and in $C_2$ agent $i_3$ receives item $c$ and $c$ is the only item that agent $p_{n-3}$ and $i_3$ demand. This implies that the algorithm cannot execute cycle $C_1$ as in this case the utility of agent $i_1$ would be $1$ and he can increase his utility by not reporting that he demands item $p'_1$. Thus, the algorithm has to execute first $C_2$ and then $C_3$ which accumulates to a total of $4$ exchanges where the optimal execution performs $n-2$ exchanges. \footnote{Observe that the problem here is that because of Pareto Efficiency the algorithm has to execute in the first instance both $C_2$ and $C_3$. In comparison,  an optimal static execution will only execute one cycle in this instance and hence agent $i_i$ would not be able to benefit by misreporting in the modified instance.}
	\end{proof}
\full{\subsection{Computational Hardness}
%TODO: Add here something about inappeoximability result
In this section we discuss the problem of computing an optimal execution from a purely computational perspective. 
We show that computing the optimal execution is NP-hard by reducing from the known NP-Complete problem of 3D-matching \footnote{Abraham et al. \cite{abraham2007clearing} also reduce from 3D-matching, however, our reductions are inherently different. Specifically, the hardness in \cite{abraham2007clearing} stems from limiting the size of the cycles whereas we have no such limitation.} :
\begin{definition} [3D-matching]
Let $X, Y,$ and $Z$ be finite, disjoint sets, of size $n$ ($|X|=|Y|=|Z|=n$) and let $T$ be a subset of $X \times Y \times Z$. Does there exist a subset $S \subseteq  T$ of size $n$ such that for any two distinct triplets $(x_1, y_1, z_1), (x_2, y_2, z_2) \in S$, we have $x_1 \ne x_2, y_1 \ne y_2,$ and $z_1 \ne z_2$?
\end{definition}	

It is known that the 3D-matching problem is also hard to approximate within a small constant factor \cite{kann1991maximum,petrank1994hardness}. In \full{Appendix \ref{app:proofs-4}}\wine{the full version} we extend our reduction to show that the problem of computing an optimal execution is also hard to to approximate within a small constant factor (i.e., the problem is APX-hard).

%The reduction is provide in Appendix \ref{app:proofs-4} showing that:
\begin{theorem}\label{thm:nphard}
	The problem of computing an optimal execution is NP-hard.
\end{theorem}
\noindent \emph{Proof.}
Recall that computing the optimal static execution can be done in polynomial time. This means that to show hardness for computing the optimal dynamic execution requires us to devise a very careful reduction in which the optimal execution has to execute specific cycles at the first round to be able to execute other cycles in the next round. Our construction is defined as follows: given an instance $X,Y,Z$ and $T\subseteq X \times Y \times Z$  of the 3D-matching problem, we will construct the following instance for the problem of computing an optimal execution:

\begin{enumerate}
	\item For every $(x,y,z)\in T$, we have agents $(xyz)_l\in N$ for~ $l\in\{2,4,6\}$ and
	items $(xyz)_l\in M$ for~$l\in\{1,3,5,7\}$ such that $S_{(xyz)_2} = \{(xyz)_1\},~ D_{(xyz)_2} = \{(xyz)_3,(xyz)_7 \}$,~
	$S_{(xyz)_4} = \{(xyz)_3\}$,~ $D_{(xyz)_4} = \{y_1\}$,~
	$S_{(xyz)_6} = \{(xyz)_5\}, D_{(xyz)_6} = \{(xyz)_7\}$.
	\item
	For every $y\in Y$, we have an item $y_1\in M$ and an agent $y_2 \in N$, such that  
	$S_{y_2}=\{ y_1\}
	$
	and $D_{y_2}=\{(xyz)_5| (x,y,z)\in T \}$.
	\item
	For every $z\in Z$, we have an agent $z_1\in N$ and an item $z_2\in M$, such that 
	$D_{z_1} = \{z_2\}
	$
	and $S_{z_1}=\{(xyz)_7| (x,y,z)\in T \}$.
	\item
	For every $x\in X$, we have item $x_2\in M$ and agents $x_1,x_3\in N$, such that $D_{x_1}= \{ x_2\}$, $S_{x_1} =\{ z_2 | z \in Z\}$,~$D_{x_3}=\{(xyz)_1| (x,y,z)\in T \}$ and $S_{x_3} = \{ x_2\}$.

\end{enumerate}
We {denote }by $G=(N,M,E)$ the trading graph for instance $\mathcal G$. To reason about the possible executions in the trading graph it will be useful to partition the demand edges in the trading graph to $8$ disjoint subsets:
\begin{equation*}
\begin{split}
&E_0 = \{ (x_1,x_2) | x\in X\}, \\
&E_1 = \{ (x_3,(xyz)_1) | x\in X, (x,y,z)\in T\}, \\
&E_2 = \{ ((xyz)_2,(xyz)_3) | (x,y,z)\in T\}, \\ 
&E_3 = \{ ((xyz)_4,y_1) | (x,y,z)\in T, y \in Y\}
\end{split}
\quad\quad
\begin{split}
&E_4 = \{ (y_2,(xyz)_5) | (x,y,z)\in T, y \in Y\} \\
&E_5 = \{ ((xyz)_6,(xyz)_7) | (x,y,z)\in T\} \\
&E_6 = \{ ((xyz)_2,(xyz)_7) | (x,y,z)\in T\} \\
&E_7 = \{ (z_1,z_2) | z \in Z\} 
\end{split}
\end{equation*}

\begin{figure}[htb!]
	\centering
	\includegraphics[width=.8\textwidth]{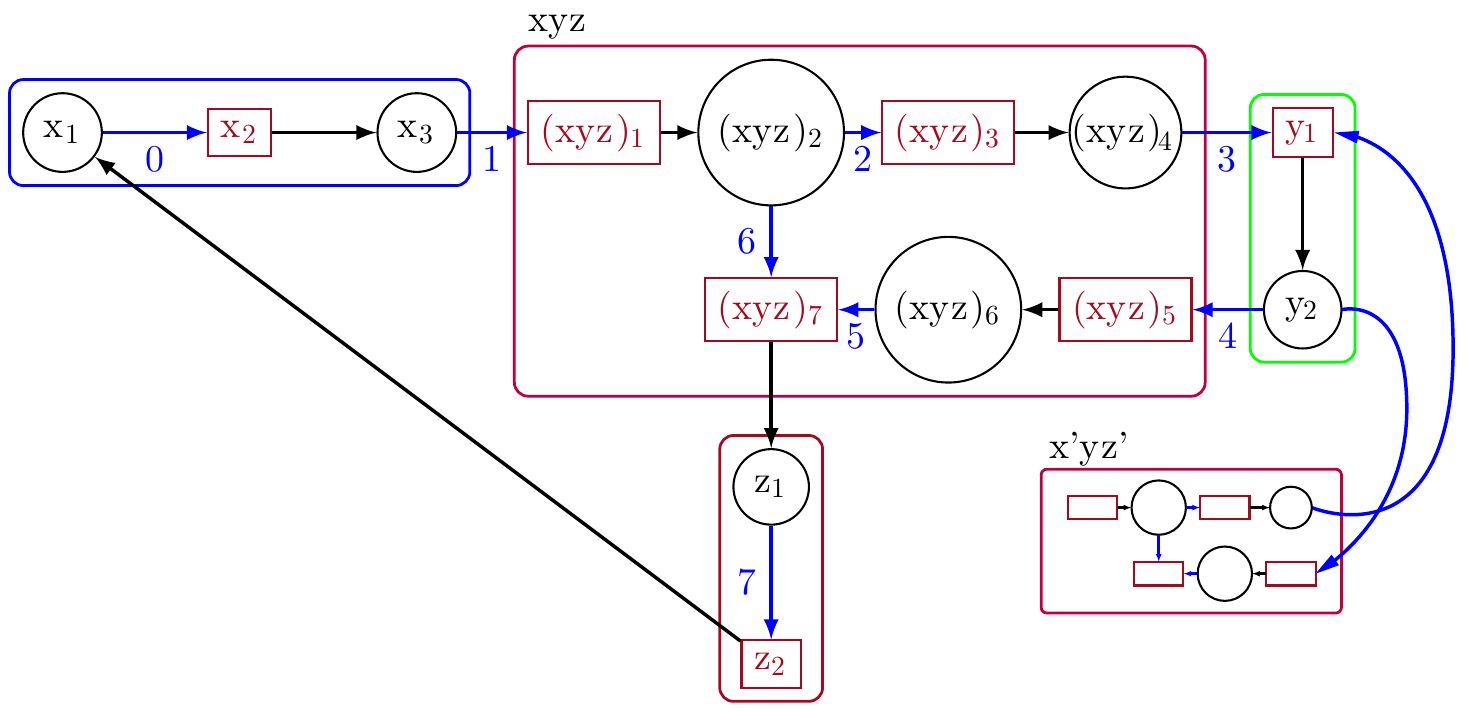}

	\caption{A small part of the trading graph $G$ constructed according to the reduction from 3D-matching. The figure includes the component of $(x,y,z)${$\in T$} and the component of $(x',y,z')\in T$. %Note that we construct seven $(xyz)_i$ nodes for each $(x,y,z) \in T$ and each of the seven nodes $(xyz)_i$ nodes component are not connected between them, but through $x_3,y_1,y_2,z_1$ nodes.
	}
	\label{fig:components-G}
\end{figure}

It is easy to see that the reduction is polynomial. In Proposition \ref{perfect-matching-8n} (below) we prove the correctness of the reduction by showing that an algorithm for computing an optimal execution in the trading graph $G=(N,M,E)$ will return an execution of size $8n$ if and only if there exists a perfect 3D-matching in $(X,Y,Z,T)$.

In order to prove the correctness of the reduction we will need several definitions and auxiliary lemmas discussing the possible types of cycles that may be executed on the trading graph $G$. We first highlight several types of cycles: 

\begin{definition}
	%Consider an execution (i.e., a sequence of cycles) on the trading graph $G$.
	We define the following types of cycles in the trading graph $G$:
	\begin{itemize} 
				\setlength\itemsep{-0.3em}
		\item A \emph {small $X$-cycle} is a cycle that includes an edge from $E_0$ and an edge from $E_6$.
		\item A \emph {large $X$-cycle} is a cycle that includes an edge from $E_0$ and does not include an edge from $E_6$.
		\item An $X$-cycle is \emph{simple} if it includes exactly a single edge from $E_0$.
		\item A \emph{$Y$-cycle} is a cycle that does not include any edge from $E_0$.
	\end{itemize}
	% A cycle that crosses an edge $(x_1,x_2) \in E_0$ and an edge in $E_6$ is referred to as \emph {a small $X$-cycle}, any other cycle that crosses an edge in $E_0$ is referred as {a large $X$-cycle}. An $X$-cycle is \emph{simple} if it includes exactly a single edge $(x_1,x_2) \in E_0$. We refer to any cycle that does not cross an edge in $E_0$ as a $Y$-cycle.
\end{definition}

The following useful lemma tells us that we can focus on executions that only includes simple $X$-cycles:
\begin{lemma}\label{lem-simpleCycles}
	Any execution $r$ can be transformed into an execution $r'$ such that $|r'|=|r|$ and in $r'$ all the $X$ cycles are simple (i.e., use a single edge from $E_0$).
\end{lemma}
\begin{proof}
	Let $C$ be a cycle step in an execution $r$ that uses the edges $(x_1,x_2),(x'_1,x'_2)\in E_0$ ($C$ may include more edges from $E_0$). Notice that such a cycle has the general form of 
	$$C= (x_1,x_2,\cdots,z_1,z_2,x'_1,x'_2,\ldots,z'_1,z'_2,x_1).$$ 
	This means that we can break $C$ into two cycles as follows: $C_1= (x_1,x_2,\cdots,z_1,z_2,x_1)$ and $C_2=(x'_1,x'_2,\ldots,z'_1,z'_2,x'_1)$. Notice that to break the cycle we used the edge $(z_2,x_1)$ instead of the edge $(z_2,x'_1)$ and the edge $(z'_2,x'_1)$ instead of the edge $(z'_2,x_1)$. As for any $x\in X$ and $z \in Z$ the edge $(z_2,x_1)$ is in the graph we have that the new edges indeed exists in the graph. Furthermore, we note that the edges $(z_2,x_1)$ and $(z'_2,x'_1)$ cannot appear in later executions since both $z_2$ and $z'_2$ only have a single incoming edge and after this edge was used as part of $C$ they cannot participate in any other cycle. Thus, we conclude that the number of exchanges in the execution that includes cycles $C_1$ and $C_2$ instead of cycle $C$ is the same as in execution $r$. We can continue in the same manner until every we reach an execution $r'$ in which each $X$-cycle includes exactly a single edge {in $E_0$}.
\end{proof} 

The restriction to executions that only includes simple $X$-cycles allow us to explicitly pinpoint the type and number of edges that participate in the different types of cycles. In \full{Appendix \ref{applem:cycles-length}}\wine{the full version} we show that:
\begin{lemma} \label{lem:cycles-length}
	The different type of cycles that can be executed are characterized as follows:	
	\begin{itemize}
		\setlength\itemsep{-0.3em}
		\item Any simple small $X$-cycle includes exactly $4$ demand edges: one edge from each of the subsets $E_0,E_1,E_6$ and $E_7$.
		\item Any simple large $X$-cycle includes exactly $7$ demand edges: one edge from each of the subsets $E_0,E_1,E_2,E_3,E_4,E_5$ and $E_7$. 
		\item 	Any $Y$-cycle includes exactly $4$ demand edges: one edge from each of the subsets $E_2,E_3,E_4,E_5$. In addition it includes one flipped edge $((xyz)_7,(xyz)_2)$ such that $((xyz)_2,(xyz)_7)\in E_6$. 
	\end{itemize}
\end{lemma} 

We now bound the number of exchanges in an execution as a function of the number of simple large $X$-cycles it includes. As a corollary we will have that any execution executes at most $8n$ exchanges:
\begin{claim}\label{claim:rle8n}
	Consider an execution $r$ on the trading graph $G$ that includes $k$ large $X$-cycles and all the $X$-cycles are simple, then $|r| \leq 8n-k$. Furthermore, $|r|=8n-k$ if and only if the number of $Y$-cycles that are executed in $r$ is $n-k$. 
\end{claim}
\begin{proof}
	By Lemma \ref{lem:cycles-length} we have that each simple large $X$-cycle contributes $7$ exchanges while any simple small $X$ cycle contributes $4$ exchanges. Thus, the total number of exchanges that are executed in $X$-cycles is at most $7k + 4(n-k) = 4n+3k.$ 
	
	Next, we consider the contribution of $Y$-cycles. By Lemma \ref{lem:cycles-length} a $Y$-cycle that includes an edge $((xyz)_7,(xyz)_2)$ can only be executed after a small $X$-cycle that includes the edge  $((xyz)_2,(xyz)_7)\in E_6$ was executed. Since each small $X$-cycle contains exactly a single edge $((xyz)_2,(xyz)_7)\in E_6$ the number of $Y$-cycles that can be executed is at most $(n-k)$. The number of exchanges in each $Y$-cycle is $4$, thus the total number of exchanges that are done in $Y$-cycles is at most $4(n-k)$. As each cycle that is executed is either an $X$-cycle or a $Y$-cycle for any execution $r$ we have that $|r|\leq 8n-k$ and this is tight if and only if the number of $Y$-cycles that are executed is $(n-k)$.
\end{proof}

Notice that Lemma \ref{lem-simpleCycles} assures us that any execution $r$ can be transformed to an execution $r'$ such that $|r|=|r'|$ and $r'$ only has simple $X$-cycles. Thus we conclude that:
\begin{corollary} \label{cor-max-8n}
	For any execution $r$, $|r| \leq 8n$.  Moreover, if $r$ is an execution that contains only simple $X$-cycles then $|r|=8n$ implies that $r$ consists of $n$ small $X$-cycles and $n$ $Y$-cycles.
\end{corollary} 
%A second useful corollary is the following:
%\begin{corollary}\label{corollary-Xcycles}
%	If $|r|=8n$ is an execution that contains only simple $X$-cycles then it consists of $n$ small $X$-cycles and $n$ $Y$-cycles.
%\end{corollary}
Finally, we are ready to prove the main proposition: 
\begin{proposition} \label{perfect-matching-8n}
	An algorithm for computing an optimal execution on the trading graph $G$ will return an execution $r$ such that $|r|=8n$ if and only if there exists a perfect matching $S\subseteq T$ for the 3D-matching instance.
\end{proposition}
\begin{proof}
	First, assume that there exists a perfect matching $S\subseteq T$, we show how to construct an execution $r$ such that $|r|=8n$. Since by Corollary \ref{cor-max-8n} we have that the maximum number of exchanges in an execution is $8n$ this implies that in this case the algorithm will return some execution $r'$ such that $|r'|=8n$. 
We now construct the execution $r$: for each $(x,y,z)\in S$ we first execute the $X$-cycle, 
$(x_1,x_2,x_3,(xyz)_1,(xyz)_2,(xyz)_7,z_1,z_2,x_1)$
and then use the flipped edge $((xyz)_7,(xyz)_2)$, such that $((xyz)_2,(xyz)_7)\in E_6$, and execute the $Y$-cycle, $(y_1,y_2,(xyz)_5,(xyz)_6,(xyz)_7,(xyz)_2,(xyz)_3,(xyz)_4,y_1).$
Notice that since $S$ is a perfect matching it has $n$ triplets and since for each triplet we execute a total of $8$ exchanges, we get that $|r|=8n$. We are left with showing that $r$ is feasible. Notice that $S$ is a feasible matching and thus, each $x\in X$, $y \in Y$ and $z\in Z$ appears in precisely one triplet $(x,y,z) \in S$. Thus, all the $X$ cycles are disjoint from one another and all the $Y$ cycle are disjoint from one another. Moreover, after executing an $X$-cycle, one can execute the $Y$-cycle that contains the edge that was just flipped.
	
We now prove the other direction, assume that the algorithm returns an execution $r$ such that $|r| = 8n$. We use Lemma  \ref{lem-simpleCycles} to construct an optimal execution $r'$ in which all $X$-cycles are simple (i.e., include a single edge from $E_0$).  Now, by Corollary \ref{cor-max-8n} we have that $r'$ includes $n$ small $X$-cycles and $n$ $Y$-cycles, each $Y$ cycles is executed only after the $X$-cycle that includes $((xyz)_2,(xyz)_7) \in E_6$ is executed. To construct the matching, for each $X$-cycle that includes an edge $((xyz)_2,(xyz)_7) \in E_6$ we add $(x,y,z)$ to the matching $S$. By definition we have that $|S|=n$. We are left to show that for any $(x,y,z),(x',y',z')\in S$ we have that $x \neq x'$, $y \neq y'$ and $z \neq z'$. Notice that all the $X$-cycles are disjoint thus if the cycle $(x_1,x_2,x_3,(xyz)_1,(xyz)_2,(xyz)_7,z_1,z_2,x_1)$ is executed we cannot execute any other cycle that includes $(x_1,x_2)$ or $(z_1,z_2)$, thus we have that $x\neq x'$ and that $z \neq z'$. Finally, the $Y$-cycles are disjoint as well implying that if we executed the cycle $(y_1,y_2,(xyz)_5,(xyz)_6,(xyz)_7,(xyz)_2,(xyz)_3,(xyz)_4,y_1)$ we cannot execute any cycle with $(y_1,y_2)$, thus we have that $y \neq y'$ as required. $\qed$ \end{proof} 

%\end{proof}

}

\section{Conclusion and Discussion}
Our paper contributes to forming the mathematical foundations of barter markets. As such, the paper does not aim to provide a full modeling of a concrete market, but rather to mathematically capture some of the major challenges in designing them. We identify a central aspect of many barter markets that yet to be studied: the market may be dynamic in the sense that the same item can move from hand to hand several times. A main contribution of our paper is identifying this aspect and formally modeling it.  The second set of contributions is in a comprehensive analysis of dynamic executions. 

Our results on the approximation ratio of strategyproof mechanisms in this setting can be interpreted in two ways. First, in many cases, the approximation ratio of $l$ (the maximal number of items an agent demands) achieved by an optimal static execution is reasonable since the number of items that an agent demands does not grow with the size of the network. This gives a justification for studying static executions even in a dynamic environment such as ours.  Second, the impossibility result showing that a strategyproof mechanism cannot provide an approximation ratio better than $\approx l/2$ suggests that to increase efficiency, barter networks should include some form of money. This may explain why many barter applications indeed often involve vouchers, for example. We hope that the understanding of barter markets we gained in this paper will provide a stepping stone towards understanding markets with vouchers.

\bibliographystyle{plain}
\bibliography{bibliography}

\full{
\newpage
		\appendix
	\section{Proofs from Section 3}
\subsection{Proof of Proposition \ref{prop:matching-properties}  }	\label{app-prop312}
Recall that we need to prove that if $ f_i(X - \{j\}) > f_i(X)$, then, for the maximum weight perfect matchings ${\mathcal M}_X$ and ${\mathcal M}_{X-\{j\}}$ there exist matchings ${\mathcal M}'_X$ and ${\mathcal M}'_{X - \{j\}}$ such that:
		\begin{enumerate}
			\item ${\mathcal M}'_X$ is a valid perfect matching of the graph $H'(G((X,S'_i),\vec{x}'_{-i}))$
			(i.e., when agent $i$ reports $X$) and ${\mathcal M}'_{X - \{j\}}$ is a valid perfect matching of the graph $H'(G((X-\{j\},S'_i),\vec{x}'_{-i}))$ (i.e., when agent $i$ reports $X-\{j\}$).
			\item $w'({\mathcal M}_X) + w'({\mathcal M}_{X-\{j\}}) =w'({\mathcal M}'_X) +  w'({\mathcal M}'_{X - \{j\}})$
			\item $w'({\mathcal M}'_X)\ne w'({\mathcal M}_X)$.
		\end{enumerate}

First by Claim \ref{clm:weighted-to-unweighted} we have that ${\mathcal M}_X$ is a perfect matching in the bipartite graph $H(G) = ( E\times \{0\},E\times \{1\},E_H)$ defined in Section \ref{sec:static-poly}. Let $\tilde {\mathcal M}_{X - \{j\}}$, to be later defined, be a perfect matching of the graph $H'(G((X-\{j\},S'_i),\vec{x}'_{-i}))$
such that $w'(\tilde {\mathcal M}_{X - \{j\}}) = w'( {\mathcal M}_{X - \{j\}})$.

Consider the multi-graph $\mathcal H$ that includes all the edges of ${\mathcal M}_X$ and $\tilde {\mathcal M}_{X - \{j\}}$.  Denote by $\mathcal E$ the multi-set of the edges in $\mathcal H$. We denote by ${P}$ a path in $\mathcal H$ that starts from $((i,j),0)$ and takes alternating edges from ${\mathcal M}_X$ and $\tilde {\mathcal M}_{X - \{j\}}$. An illustration of the path can be found in Figure \ref{fig:H(G)}. We use the path $P$ to define two matchings:\footnote{The construction borrows from the proof of Lemma 2 in \cite{bar2016tight}.}
\begin{align*}
{\mathcal M}'_X = ({P}\cap {\mathcal M}_X)\cup ((\mathcal{E} -  {P})\cap \tilde {\mathcal M}_{X - \{j\}}),~
{\mathcal M}'_{X - \{j\}} = ({P}\cap \tilde {\mathcal M}_{X - \{j\}})\cup ((\mathcal{E} -  {P})\cap {\mathcal M}_X)
\end{align*}
\begin{figure}
	\centering
	\includegraphics[width=8cm]{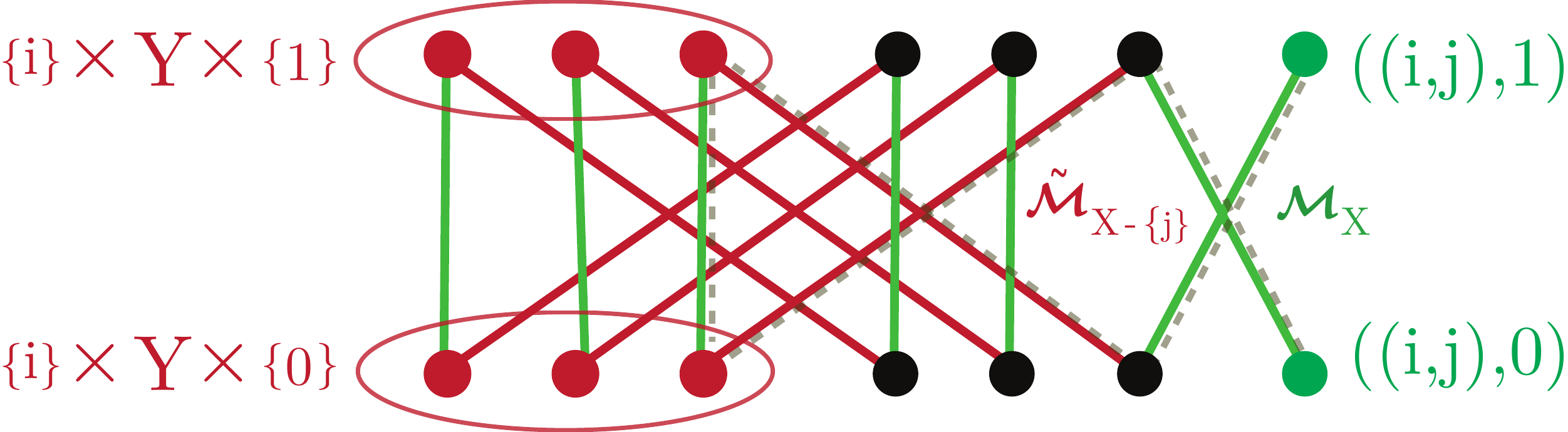}
	\caption{The multi-graph $\mathcal H$. The edges of ${\mathcal M}_X$ are colored in green and the edges of $\tilde {\mathcal M}_{X -\{j\}}$ are colored in red. The path $P$ is marked by dashed edges.
	}
	\label{fig:H(G)}
\end{figure}
By definition we have that the multiset ${\mathcal M}'_X \cup {\mathcal M}'_{X - \{j\}}$ is the same as the multiset $ {\mathcal M}_X \cup \tilde {\mathcal M}_{X - \{j\}}$. This implies that $w'({\mathcal M}'_X)+w'({\mathcal M}'_{X - \{j\}}) = w'({\mathcal M}_X)+w'(\tilde {\mathcal M}_{X - \{j\}})$. Since $w'(\tilde {\mathcal M}_{X - \{j\}}) = w'({\mathcal M}_{X - \{j\}})$ this proves statement $2$ of Proposition \ref{prop:matching-properties}.

We now prove statement 1 of Proposition \ref{prop:matching-properties} showing that ${\mathcal M}'_X$ and ${\mathcal M}'_{X - \{j\}}$ are perfect matchings of $H'(G((X,S'_i),\vec{x}'_{-i}))$ and $H'(G((X-\{j\},S'_i),\vec{x}'_{-i}))$ respectively. Let $V=E\times \{0,1\}$ and $V_{-j} = (E -  {\{(i,j)\}})\times \{0,1\}$. We need to show that ${\mathcal M}'_X$ matches every node in $V$ to exactly one node and that ${\mathcal M}'_{X-\{j\}}$ matches every node in $V_{-j}$ to exactly one node. We separate the proof to two parts:

\begin{itemize}
	\item ${\mathcal M}'_X$ and ${\mathcal M}'_{X - \{j\}}$ match all the nodes in $V_{-j}$ -- Consider a node $a\in V_{-j}$.  Note that $a$ is matched both in 
	${\mathcal M}_X$ and $\tilde {\mathcal M}_{X - \{j\}}$. By construction if $P$ visits $a$ then it will be matched in ${\mathcal M}'_X$ to the same node as in ${\mathcal M}_x$ and in  ${\mathcal M}'_{X-\{j\}}$ to the same node as in $\tilde {\mathcal M}_{X - \{j\}}$. Otherwise, it will be matched in ${\mathcal M}'_X$ to the same node as in $\tilde {\mathcal M}_{X - \{j\}}$ and in ${\mathcal M}'_{X-\{j\}}$ to the same node as in ${\mathcal M}_x$ .
	
	\item ${\mathcal M}'_X$ matches $((i,j),0)$ and $((i,j),1)$ and ${\mathcal M}'_{X - \{j\}}$ does not -- Recall that $((i,j),0)$ and $((i,j),1)$ are only matched in ${\mathcal M}_X$. By definition, we have that ${P}$ begins with the node $((i,j),0)$ implying that only ${\mathcal M}'_X$ matches $((i,j),0)$. To complete the proof we should show that $((i,j),1)$ is part of ${P}$ which implies that $((i,j),1)$ is also only matched in ${\mathcal M}'_X$. In particular, we claim that the path ${P}$ ends with $((i,j),1)$. Notice that in the multi-graph $\mathcal H$ the degree of each node in $V_{-j}$ is $2$ since it is matched both in ${\mathcal M}_X$ and ${\mathcal M}_{X - \{j\}}$ and the degree of $((i,j),0)$ and $((i,j),1)$ is $1$. This implies that a path that starts with $((i,j),0)$ can only end with $((i,j),1)$.
\end{itemize}

Lastly, we prove statement 3 of Proposition \ref{prop:matching-properties} showing that $w'({\mathcal M}'_X)\ne w'({\mathcal M}_X)$. To this end we will use the following claim to pick a perfect matching $\tilde {\mathcal M}_{X - \{j\}}$ of the graph $H'(G((X-\{j\},S'_i),\vec{x}'_{-i}))$ such that $w'(\tilde {\mathcal M}_{X - \{j\}}) = w'( {\mathcal M}_{X - \{j\}})$ and ${P}$ includes a single edge $((i,y),0)$ such that $y \in Y$.
% \footnote{Observe that since $(y,0)$ and $(y,1)$ are connected by an edge in ${\mathcal M}_X$ any path that goes through $(y,0)$ has to go through $(y,1)$ and vice versa.}  

%TODO: Next version, move this definition earlier
\begin{claim}\label{clm:px-one-Y}
	Let $G_{X-\{j\}} = {G((X-\{j\},S'_i),\vec{x'}_{-i})}$. For any perfect matchings ${\mathcal M}_X$ of the bipartite graph $H'(G((X,S'_i),\vec{x'}_{-i}))$ and matching ${\mathcal M}_{X -\{j\}}$ of the bipartite graph $H'(G_{X-\{j\}})$, there exists a matching $\tilde {\mathcal M}_{X -\{j\}}$ such that:
	\begin{itemize}
		\item $\tilde {\mathcal M}_{X -\{j\}}$ is a perfect matching of the graph $H'(G_{X-\{j\}})$ and $w'(\tilde {\mathcal M}_{X - \{j\}}) =  w'({\mathcal M}_{X -\{j\}})$.
		\item In the multi-graph ${\mathcal H}$ the path ${P}$ that begins with $((i,j),0)$ and takes alternating edges from ${\mathcal M}_X$ and $\tilde {\mathcal M}_{X -\{j\}}$ visits at most a single node {$((i,y),0)$} such that $y\in Y$.
	\end{itemize}
\end{claim}

Before proving Claim \ref{clm:px-one-Y}, we apply it to prove statement 3 of Proposition \ref{prop:matching-properties}. Recall that $|Y|\geq 2$, the fact that there is a single item $y \in Y$ such that ${P}$ includes edges from $((i,y),0)$ implies that there exists another item $y'\in Y$ such that $((i,y'),0)$ is not on the path ${P}$. Let $e\in E_H$ be the edge matching $((i,y'),0)$ in $\tilde {\mathcal M}_{X - \{j\}}$. We have that $e \in \mathcal E  -  {P}$ implying that $e$ is in ${\mathcal M}'_X$ and not in ${\mathcal M}_X$. Since $e\in E_2$ it has a positive weight and hence by the construction of the weights we have that $w'({\mathcal M}'_X)\ne w'({\mathcal M}_X)$ as required. This concludes the proof of Proposition \ref{prop:matching-properties}.

\vspace{0.1in}
\noindent \textbf{Proof of Claim \ref{clm:px-one-Y}.}
\begin{figure}[]
	\begin{center}
		\subfigure[${P}({\mathcal M}_X, {\mathcal M}_{X -\{j\}})$ is marked with dashed line]{
			\includegraphics[width=6cm]{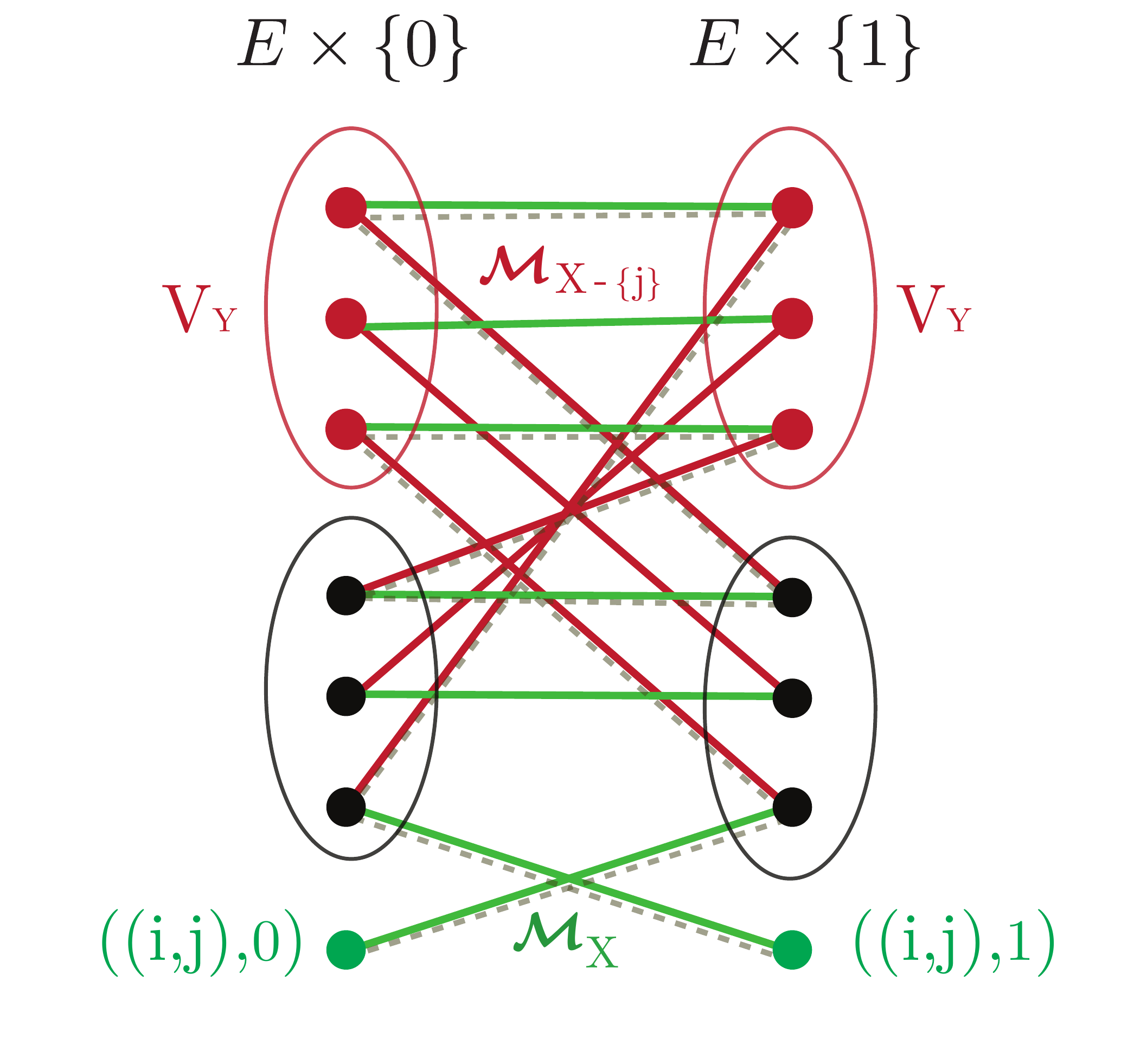}
			\label{fig:path-P}
		}\hspace{5mm}
		\subfigure[${P}({\mathcal M}_X,\tilde {\mathcal M}_{X - \{j\}})$ visits a single node $(y,0)$ such that $y\in Y$]{
			\includegraphics[width=6cm]{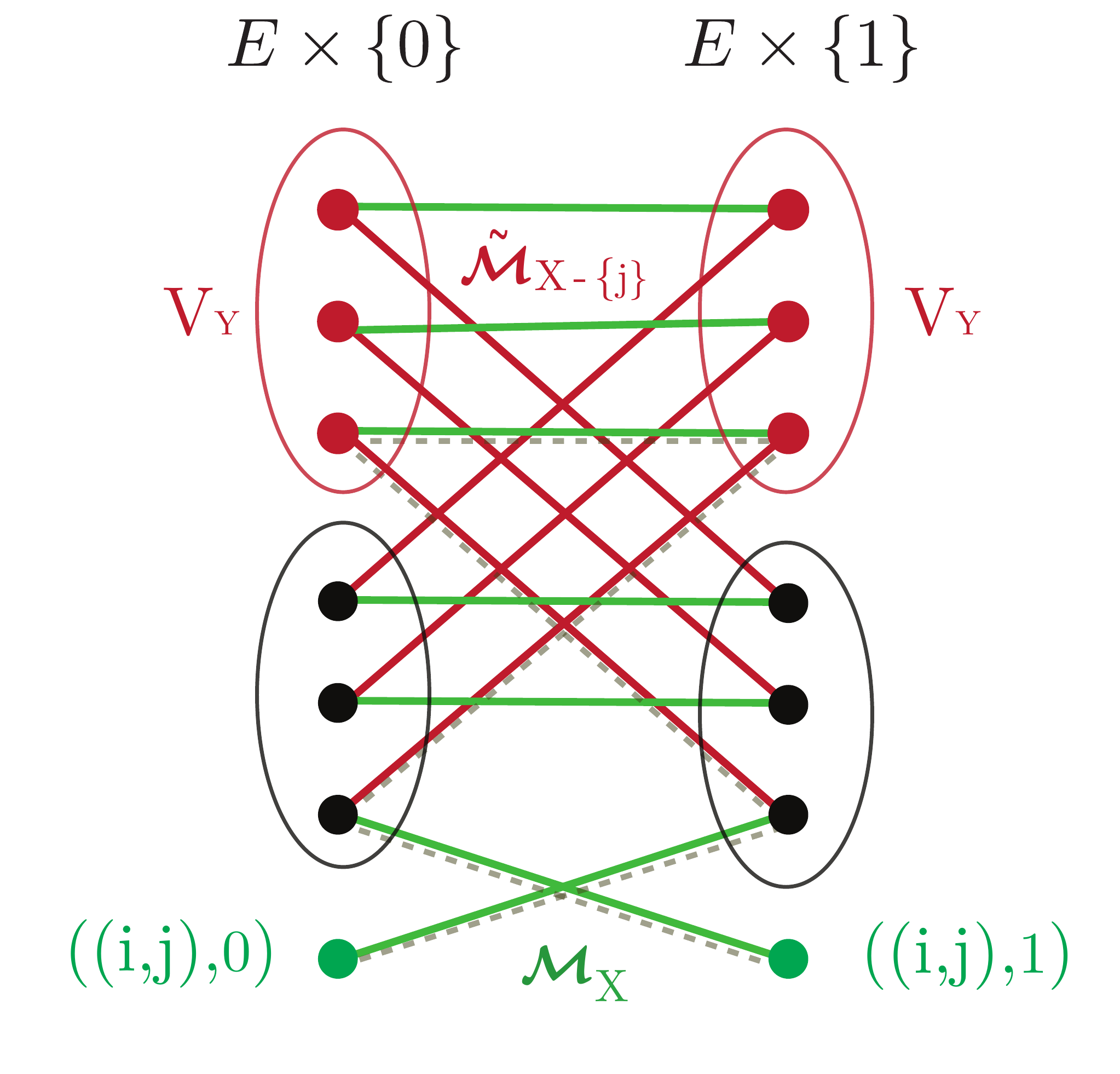}
			\label{fig:M+Mx2}
		}
		\caption{
			{\small
				An illustration of the construction in the proof of Claim \ref{clm:px-one-Y}. The edges of ${\mathcal M}_X$ are colored in green and the edges of ${\mathcal M}_{X -\{j\}}$ and $\tilde {\mathcal M}_{X -\{j\}}$ are colored in red. 
				In Figure~\ref{fig:path-P} we have a path ${P}({\mathcal M}_X, {\mathcal M}_{X -\{j\}})$ that visits more than a single node $(y,0)$ such that $y\in Y$. In Figure~\ref{fig:M+Mx2} we have the matching  $\tilde {\mathcal M}_{X - \{j\}}$ which has the same weight as ${\mathcal M}_{X - \{j\}}$ but here ${P}({\mathcal M}_X,\hat{{\mathcal M}_{X - \{j\}}})$ visits only in a single node $(y,0)$ such that $y\in Y$.
			}
			\label{fig:y-once}
		}
	\end{center}
	\vspace*{-0.2in}
\end{figure}
%\begin{proof}
For sake of clarity we will use in this proof the notation ${P}({\mathcal M}_X, {\mathcal M}_{X -\{j\}})$ to denote the alternating path starting from $((i,j),0)$ in the multi-graph defined by the union of ${\mathcal M}_X$ and ${\mathcal M}_{X -\{j\}}$. 
If ${P}({\mathcal M}_X, {\mathcal M}_{X -\{j\}})$ visits at most a single node $((i,y),0)$ such that $y\in Y$, then we are done.
Else, the path visits more than one node $((i,y),0)$ such that $y\in Y$. In Figure \ref{fig:path-P} we illustrate a path ${P}({\mathcal M}_X, {\mathcal M}_{X -\{j\}})$ that visits two nodes  $((i,y),0)$ and $((i,y'),0)$ such that $y,y'\in Y$.
Let $V_Y = \{ ((i,y),b) : y\in Y,~b\in\{0,1\} \}$. Observe that the first node $s\in V_Y$ that the path ${P}({\mathcal M}_X, {\mathcal M}_{X -\{j\}})$ visits has to be in $E\times \{0\}$ and the last node $t\in V_Y$ that ${P}({\mathcal M}_X, {\mathcal M}_{X -\{j\}})$ visits has to be in $E\times \{1\}$.  The reason for this is that since the path ${P}({\mathcal M}_X, {\mathcal M}_{X -\{j\}})$ begins with an edge from $\mathcal M_X$ that connects $((i,j),0)$ to some node $((j,k),1)$ and takes alternating edges from $\mathcal M_X$ and ${\mathcal M}_{X -\{j\}}$ then all the edges on the path that enter nodes in $E\times\{1\}$ are from $\mathcal M_X$. The fact that $\forall y\in Y$ we have that $(((i,y),0),((i,y),1))\in \mathcal M_X$ implies that the first node $s\in V_Y$ that ${P}({\mathcal M}_X, {\mathcal M}_{X -\{j\}})$ visits has to be in $E\times \{0\}$. For the same reasons we have that the last node $t \in V_y$ that the path visits has to be in $E\times \{1\}$. Let $s=((i,k),0)$ and $t=((i,u),1)$.

Denote by $((k',i),0)$ the node that was matched to $((i,k),1)$ and by $((u',i),0)$ the node that was matched to $((i,u),1)$ by ${\mathcal M}_{X -\{j\}}$. We construct the matching $\tilde {\mathcal M}_{X -\{j\}}$ that is identical to ${\mathcal M}_{X -\{j\}}$ except for matching  $((i,k),1)$ to $((u',i),0)$ and $((i,u),1)$ to $((k',i),0)$. Observe that the only nodes that 
${P}({\mathcal M}_X,\tilde {\mathcal M}_{X -\{j\}})$ visits from $V_y$ are  $((i,k),0)$ and $((i,k),1)$. The reason for this is that after ${P}({\mathcal M}_X,\tilde {\mathcal M}_{X -\{j\}})$ visits $((i,k),0)$ it continues to $((i,k),1)$ and then visits $((u',i),0)$. By construction we have that after this node the path does not visit any node $a\in V_Y$. An illustration of this construction can be found in Figure \ref{fig:y-once}.

To show that $\tilde {\mathcal M}_{X -\{j\}}$ is a valid perfect matching we need to show that the edges $((u',i),0),((i,u),1)$ and $((u',i),0),((i,k),1)$ are part of the graph $H'(G_{X-\{j\}})$.  Recall that both $(i,k)$ and $(i,u)$ are demand edges of the same agent $i$. The edges $(u',i)$ and $(k',i)$ are two supply edges of agent $i$. This means that in the trading graph $G_{X-\{j\}}$ the edges $(u',i)$ and $(i,k)$ are adjacent and the edges $(k',i)$ and $(i,u)$ are adjacent. Hence, the edges $((u',i),0),((i,k),1)$ and $((k',i),0),((i,u),1)$ appear in $H'(G_{X-\{j\}})$. Finally, notice that $w'(\tilde {\mathcal M}_{X - \{j\}}) =  w'({\mathcal M}_{X -\{j\}})$ since $w'(((k',i),0),((i,k),1)) =  w'(((k',i),0),((i,u),1))$ and $w'(((u',i),0),((i,u),1)) =  w'(((u',i),0),((i,k),1))$.
$\qed$
\vspace{0.1in}

	\section{Proofs from Section 4} \label{app:proofs-4}
\subsection{Proof of Lemma \ref{lem:cycles-length}} \label{applem:cycles-length}
For clarity we partition the proof to two lemmas:
	\begin{lemma}
		For simple $X$-cycles we have that:
		\begin{itemize}
			\item Any simple small $X$-cycle includes exactly $4$ demand edges: one edge from each of the subsets $E_0,E_1,E_6$ and $E_7$.
			\item Any simple large $X$-cycle includes exactly $7$ demand edges: one edge from each of the subsets $E_0,E_1,E_2,E_3,E_4,E_5$ and $E_7$. 
		\end{itemize}
	\end{lemma}
	\begin{proof}
		Let $C$ be a simple $X$-cycle which includes the edge $(x_1,x_2)\in E_0$. Observe that $C$ includes one {edge }from each of the sets $E_0,E_1$ and $E_7$. This is due to the fact that $C$ has exactly one edge from $E_0$ and the only path in the graph that includes an edge from $E_0$ has an edge from $E_1$ and an edge from $E_7$. Furthermore, $C$ cannot include more than one edge from $E_7$ since each such edge has to be followed by an edge from $E_0$. For similar reasons $C$ cannot include more than a single edge from $E_1$.
		
		The observation that $C$ visits exactly a single edge from the sets $E_0,E_1$ and $E_7$ implies that $C= (( x'  y' z')_7, z'_1, z'_2,x_1,x_2,x_3,(xy z)_1,(xy z)_2) \circ p'$. 
		If $C$ is a small cycle then we know it includes an edge from $E_6$ since only edges from $E_7$ can follow an edge from $E_6$ \footnote{Note that in principle an edge from $E_5$ can be flipped an then follow an edge from $E_6$ however in this case after taking the edge a node without any outgoing edges will be reached.} we get that in this case $$C= (( x'  y'  z')_7, z'_1, z'_2,x_1,x_2,x_3,(xy z)_1,(xy z)_2,(x y z)_7)$$ and $(x y z) = ( x'y'z')$ which completes the proof of the first statement. 
		
		Consider the case that $C$ is a large $X$-cycle. We will show that the path $p'$ includes exactly one edge from each of the sets $E_2,E_3,E_4,E_5$. Since $p'$ starts from $(xyz)_2$ and does not include any edges from $E_6$ it has to continue to an edge from $E_2$ the only {demand} edge that follows an edge from $E_2$ is an edge from $E_3$. At this point the path reaches some node $y_1$, the next demand edge that the path must take is from $E_4$\footnote{In principle the path may take an edge from $E_3$ that was flipped to get from $y'_1$ to some node $(xyz)_4$ however since in this case $(xyz)_4$ has no outgoing edges that path will end at $(xyz)_4$.} Finally, after taking an edge from $E_4$ the path will have to take an edge from $E_5$ and reach $(xyz)_7$. From a similar reason as we described for small $X$-cycles at this point the path will end. Therefore we conclude that any simple large $X$-cycle includes exactly $7$ demand edges: one edge from each of the subsets $E_0,E_1,E_2,E_3,E_4,E_5$ and $E_7$. 
	\end{proof}
	
	\begin{lemma}
		Any $Y$-cycle includes exactly $4$ demand edges: one edge from each of the subsets $E_2,E_3,E_4,E_5$. In addition it includes one flipped edge $((xyz)_7,(xyz)_2)$ such that $((xyz)_2,(xyz)_7)\in E_6$. 
	\end{lemma}
	\begin{proof}
		We will show that any cycle that does not cross an edge $(x_1,x_2)\in E_0$  uses exactly 4 demand edges, one from each of the sets $E_2,E_3,E_4,E_5$, and uses one flipped edge $((xyz)_7,(xyz)_2)$ such that $((xyz)_2,(xyz)_7)\in E_6$. To prove this we will prune the graph sequentially. First, we remove all edges $(x_1,x_2)\in E_0$. This implies that we can also remove all edges $(x_2,x_3)$ such that $x\in X$ as the only incoming edge to $x_3$ is from $x_2$. Similarly, we can remove all edges $(z_2,x_1)$ such that $z\in Z$ and $x\in X$ since the only outgoing edges from $z_2$ is to $x_1$ such that $x \in X$. Now, we can also remove all edges $(z_1,z_2)\in E_7$ since the only outgoing edge from $z_1$ is to $z_2$. From the same reasons we can remove all edges $(x_3,(xyz)_1)$ for $x\in X$ and $(x,y,z)\in T$, {followed by removing $((xyz)_1,(xyz)_2)$ for $(x,y,z)\in T$ } and removing all the edges $((xyz)_7,z_1)$ such that $(x,y,z)\in T$ and $z\in Z$. After these pruning steps our graph is composed from the components in Figure \ref{fig:xyz-component} such that components for $(x,y,z)\in T$ and $(x',y,z')\in T$ are connected to one another by nodes $y_1$ and $y_2$ such that $y\in Y$.

		\begin{figure}[htbp!]
			\centering
			\includegraphics[width=0.5\textwidth]{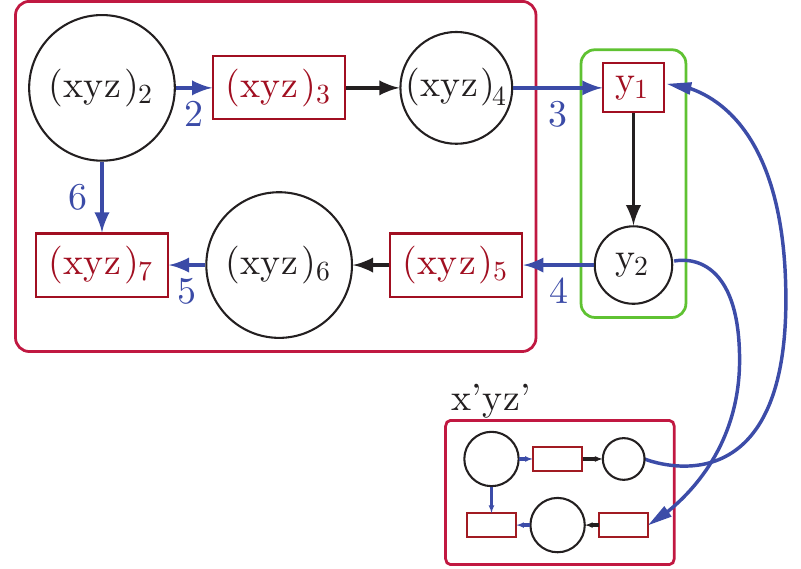}
			\caption{The trading graph after we prune it. The
				$\{(xyz)_i:2\le i \le 7\}$ components are connected only through nodes $y_1$ and $y_2$.}
			\label{fig:xyz-component}
		\end{figure}
		Observe that in the remaining graph there is no cycles that does not use the edge $(y_1,y_2)$ since any component for $(x,y,z)\in T$ is connected only to another component for $(x',y,z')\in T$, through $y_1,y_2$. Let $C$ be a cycle which uses the edge $(y_1,y_2)$, notice that any path that starts from $y_2$ has to first take an edge $(y_2,(xyz)_5)\in E_4$ then it has no choice but to take an edge $((xyz)_5,(xyz)_6)$ and then take an edge $((xyz)_6,(xyz)_7) \in E_5$. Once the path reached $(xyz)_7$ the only option to continue is to use the flipped edge $((xyz)_7,(xyz)_2)$ otherwise the path reaches a dead end. From this point it has to continue through edges $((xyz)_2,(xyz)_3)\in E_2$ ,$((xyz)_3,(xyz)_4)$ and $((xyz)_4,y_1)\in E_3$ and close a cycle 
		by reaching back to $y_1$.
	\end{proof}

\subsection{Hardness of Approximation}
We will use the same reduction as in the proof of Theorem \ref{thm:nphard} to show the hardness of computing an approximate execution.
\begin{theorem}\label{THM:no-ptas}
	There exists some constant $c>1$ such that there is no polynomial time $c$-approximation algorithm for computing an optimal execution, unless $P=NP$.
\end{theorem}
\begin{proof}
	It is known (\cite{petrank1994hardness}) that for the problem of 3D-matching it is NP-hard to distinguish between an instance that has a maximum matching of size $n$ and an instance that has a maximum matching of size $n(1-\eps)$ for some small constant $\eps$ . For $\alpha>1$, consider an $\alpha$-approximation for computing an optimal execution. If the 3D-matching instance admits a perfect matching then by Proposition \ref{perfect-matching-8n} the optimal execution in the trading graph $G$ is $r_o$ such that $|r_o|=8n$ and hence the approximation algorithm should return an execution of size at least $\frac{8n}{\alpha}$.
	On the other hand, in case that the maximum size of the matching is $(1-\eps)n$ by Lemma \ref{Lemma:at-most-k} we have that the optimal execution is $r_o$ such that $|r_o|\leq 7n+(1-\eps)n =8n-\eps n$. Notice that it has to be the case that $8n-\eps n \ge\frac{8n}{\alpha}$ as otherwise our approximation algorithm could be used to separate between the case that the size of the maximum matching is $n$ and the case that the size of the maximum matching is $(1-\eps)n$. The fact that $8n-\eps n \geq \frac{8n}{\alpha}$ implies that we cannot get approximation ratio $\alpha < 1+\frac{\eps}{8-\eps} $.
\end{proof}

\begin{lemma}\label{Lemma:at-most-k} If the maximal 3D-matching is of size $d$ then for any execution $r$ for the trading graph $G$, we have that $|r|\leq 7n+d$. 
\end{lemma}
\begin{proof}
	In the proof of Proposition \ref{perfect-matching-8n} we showed how to construct a perfect matching for an execution that has $n$ small $X$-cycles and $n$ $Y$ cycles and on the other direction showed that if there exists a perfect matching then there exists an execution that has $n$ small $X$-cycles and $n$ $Y$ cycles. Similarly, one can show that there exists a matching of size at $t$ if and only if there exists an execution that has at least $t$ small $X$-cycles and $t$ $Y$ cycles. This implies that if the size of the maximal 3D-matching is $d$ then, there is no execution that includes $d_X$ small $X$ cycles and $d_Y$ $Y$ cycles such that $d_X> d$ and $d_Y > d$. Furthermore, observe that $d_X \geq d_Y$ as a $Y$-cycle with an edge $((xyz)_7,(xyz)_2)$ can only be executed after the $X$-cycle including the edge $((xyz)_2,(xyz)_7)\in E_6$ was executed. In particular this implies that $d_Y\leq d$. Thus, we have that overall the total number of exchanges in any execution is : $|r| \leq 4d_X+4d_Y + 7(n-d_X) = 7n-3d_X+4d_Y \leq 7n +d_Y \leq 7n+d$.
	
\end{proof}

	\section{On the Approximation Ratio of the Greedy Algorithm} \label{sec-greedy}
We consider the simple greedy dynamic algorithm that at each step computes the optimal static execution, executes it and updates the graph accordingly. We show that in the worst case, this algorithm cannot guarantee a better approximation ratio than the algorithm that computes the optimal static execution. We show an instance 
$\mathcal G$, in which after we execute the optimal static execution we can no longer do any exchanges. 
However, the optimal dynamic execution can perform considerably more exchanges. We construct a 
family of instances $\mathcal G_d$ parameterized by $d\geq 2$, such that as $d$ goes to infinity the 
approximation ratio of the greedy algorithm approaches $l$.
\begin{theorem}
The greedy algorithm cannot guarantee an approximation better than $l$.
\end{theorem}
\begin{proof}
Consider the following instance $\mathcal G_d$ for $d\geq 2$ illustrated in Figure \ref{fig:Greedy}. The set of agents includes $d$ subsets of cardinality $l$: $N_1\cup N_2 \cup \ldots \cup N_d$, such that for any $1 \leq y \leq d$, $N_y=\{ i^1_y,i^2_y,\ldots, i^l_y \}$. The set of items includes $d$ subsets of cardinality $l$: $M_1 \cup M_2 \cup \ldots \cup M_{d}$, such that for any $1 \leq y \leq d$,  $M_y = \{ j^1_y,j^2_y,\ldots, j^l_y \}$. The demand set and supply set of the agents is defined as follows: 
\begin{itemize}
	\item For $i^k_1 \in N_1$, $S_{i^k_1} =\{j^k_1\}$ and $D_{i^k_1} = \{ j^k_2 \}$.
	\item For $2\le y \le d-1$ and $i^k_y \in N_y$, $S_{i^k_y} =\{j^k_y\}$ and $D_{i^k_y} = M_y - \{ j^k_y \}\cup \{ j^k_{y+1}\}$.
	\item For $i^k_d \in N_d$, $S_{i^k_d} =\{j^k_d\}$, $D_{i^k_d} ={M_d - \{ j^k_d \} \cup} \{ j^{k +1 (mod~l)}_1 \} $.
\end{itemize}

\begin{figure}[htbp!]
	\centering
	\includegraphics[width=12cm]{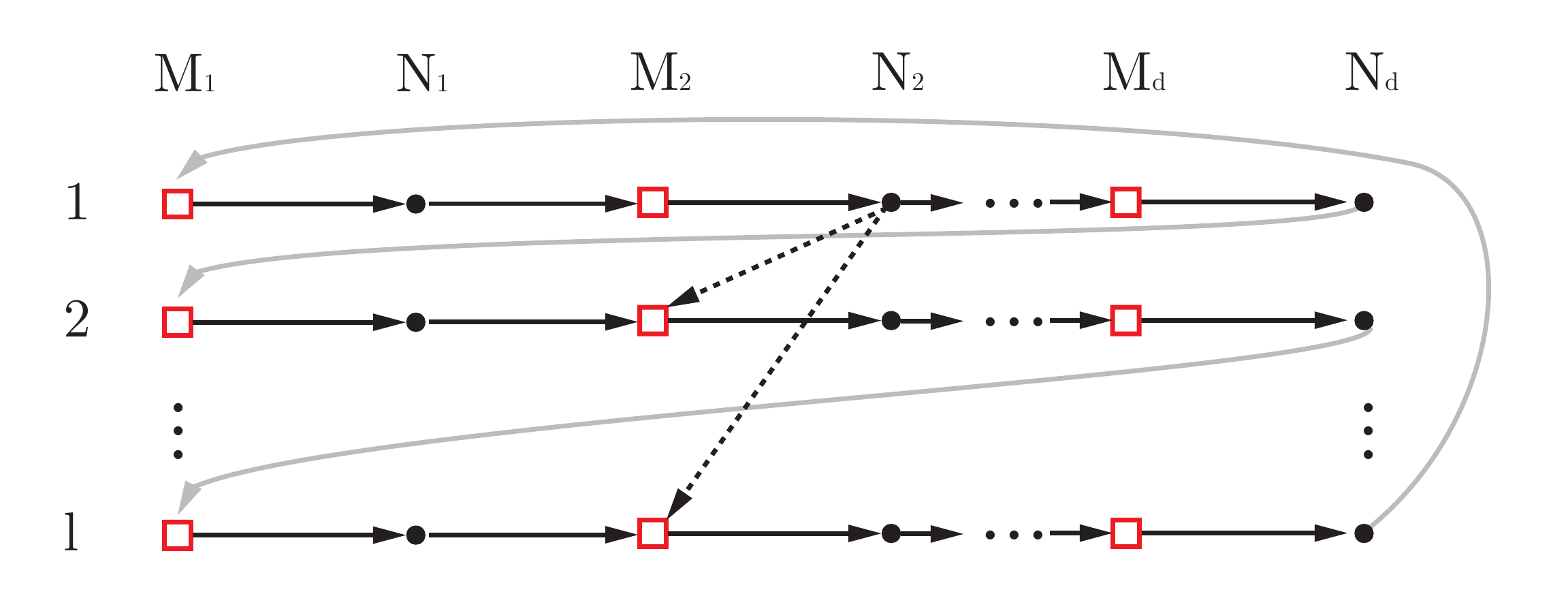}
	\caption{The trading graph for an instance showing that a greedy algorithm cannot preform better than the optimal static algorithm. For simplicity the demand edges of each agent ${i^k_y},~{2\le y \le d}$, for all the items in $M_y-\{j^k_y\}$ except for agent $i^1_2$ (the dashed edges) are omitted.}
	\label{fig:Greedy}
\end{figure}
Notice that the graph has a giant cycle in which all the agents participate:
\begin{align*}
C_{all} = (j^1_1,i^1_1,j^1_2,i^1_2,\ldots,i^1_d,j^2_1,i^2_1,\ldots,i^2_d\ldots,j^l_1,i^l_1,\ldots,i^l_d,j^1_1)
\end{align*}
It is easy to see that the optimal static execution will execute the cycle $C_{all}$. We show that after this cycle is executed there are no other exchanges that can be executed. Thus, the total number of exchanges that the greedy algorithm performs
is $l \cdot d$.

\begin{claim}\label{claim:C-all-in-G}
For any $d\geq 2$, the greedy algorithm in the instance $\mathcal G_d$ cannot perform any more exchanges after executing the cycle $C_{all}$.
\end{claim} 
\begin{proof}
We prove by induction on $1\le y \le d$, showing that after we execute $C_{all}$ all the agents in $N_y$ 
	cannot take part in another exchange:
	\begin{description}
		\item[Base case:] After we execute $C_{all}$ the agents in $N_1$ own the only item they demand and hence they cannot participate in any further exchange. 
		\item[Induction hypothesis:] For $2\le q\le d$, Assume that after $C_{all}$ was executed the agents in $N_{y-1}$ cannot take part in 
		other exchange, we will show that the agents in $N_y$ cannot take part in 
		other exchange. 
		\item[Induction step:] After $C_{all}$ was executed all the items in the set $M_y$ are owned by agents in the set $N_{y-1}$. By the induction hypothesis the agents in $N_{y-1}$ cannot take part in any more 
		exchanges. Thus,  the items in $M_y$ cannot participate in any exchange. As all the items that agents in $N_y$ demand are in $M_y$ this implies that the agents in $N_y$ also cannot participate in any further exchange.
	\end{description}
\end{proof}

Now consider the optimal execution. The optimal execution first for each component $N_y,M_y$ for $2 \leq y \leq d$ the execution executes $l-1$ cycles in which all the agents in $N_y$ and the items in $M_y$ participate. The total number of exchanges in each component is $l \cdot (l-1)$ and in all the components together is $(d-1)\cdot l \cdot (l-1)$. Now consider the trading graph after executing all of these exchanges.  We claim that the in-degree and out-degree of each 
node in the graph is exactly 1. To see why this is the case, first observe that each agent owns exactly a single item and each item is owned by a single agent. Furthermore, for $1\leq y \leq d$ any agent $i \in N_y$ only demands item $j \in M_{(y+1) (mod~d)}$ as this is the only item that he has not received yet. Finally, we notice that each item is demanded by a single agent:
 for $1 \leq d$ any item  $j^t_y \in M_y$, is demanded by agent $i^t_{(y-1)(mod~d)}$. The fact that the in-degree and out-degree o	f all the nodes is $1$ implies that all the edges of the graph can be covered by cycles. Thus, the optimal execution fulfills all the agents' demands and hence executes $l + (d-1)\cdot l^2$ exchanges.
 
 Recall that a greedy algorithm cannot perform anymore exchanges after executing the optimal static execution. Thus, the greedy algorithm performs $ d\cdot l$ exchanges and on this instance achieves an approximation ratio of $ \frac{l+ (d-1)\cdot l^2}{d\cdot l}=\frac{1+(d-1)\cdot l}{d} \rightarrow_{d\rightarrow \infty} l$.
\end{proof}
}

\end{document}